\documentclass[11pt]{amsart}                                    
\usepackage{amsmath}
\usepackage{mdframed}
\usepackage{graphicx,amssymb}
\usepackage{algorithm}
\usepackage{algorithmic}
\usepackage{tcolorbox}
\usepackage{pdfsync}
\usepackage{amsthm,color}
\usepackage{enumitem}
\usepackage{multirow}
\usepackage{float}
 \usepackage{diagbox}

\newtheorem{theorem}{Theorem}[section]
\newtheorem{lemma}[theorem]{Lemma}
\newtheorem{proposition}[theorem]{Proposition}
\newtheorem{definition}[theorem]{Definition}
\newtheorem{corollary}[theorem]{Corollary}

\newtheorem{remark}[theorem]{Remark}

\newtheorem{question}[theorem]{Question}

\newtheorem{assumption}[theorem]{Assumption}

\def\R{{\mathbb R}}
\def\Z{{\mathbb Z}}
\def\T{{\mathbb T}}

\def\Rd{{\R}^d}

\numberwithin{equation}{section}

\makeatletter
\def\BState{\State\hskip-\ALG@thistlm}
\makeatother

\begin{document}

\title{Stable Phaseless Sampling and Reconstruction of Real-Valued
 Signals with Finite Rate of Innovations} 
\author{Cheng Cheng}

\address{Cheng: Department of Mathematic, Duke University, Durham, NC, 27708, USA, and The Statistical and Applied Mathematical Sciences Institute (SAMSI), Durhma, NC, 27709, USA}
\email{cheng87@math.duke.edu}

\author{Qiyu Sun}
\thanks{This work is partially supported by
SAMSI under the National Science Foundation (DMS-1638521), and the  National Science Foundation (DMS-1412413).}

\address{Sun: Department of Mathematics,
University of Central Florida,
Orlando, Florida, 32816, USA}
\email{qiyu.sun@ucf.edu}
\maketitle

\begin{abstract}
A spatial signal  is defined by its evaluations on the whole domain.  In this paper, we consider stable reconstruction of
real-valued signals with finite rate of innovations (FRI), up to a sign, from
their magnitude measurements on the whole domain or  their phaseless samples on a discrete subset. 
FRI signals appear
    in many engineering applications such as  magnetic resonance spectrum,  ultra wide-band communication and
    electrocardiogram.
For an FRI signal, we introduce an undirected  graph to  describe its topological structure.
We establish the equivalence between the graph connectivity and phase retrievability of FRI signals,
and we apply the graph connected component decomposition  to find all FRI signals that have  the same magnitude  measurements
  as the original FRI signal has.
We construct  discrete sets with finite density  explicitly so that
 magnitude measurements of FRI signals on the whole domain are determined by
their samples taken on those  discrete subsets.
In this paper, we also propose a stable algorithm with linear complexity to reconstruct FRI signals from their phaseless samples  on the above phaseless sampling set.  The  proposed algorithm  is demonstrated theoretically and numerically  to  provide a suboptimal approximation to the original  FRI signal in magnitude measurements.
\end{abstract}

\section{Introduction}

A spatial signal  $f$ on a domain $D$ is defined by its evaluations $f(x), x\in D$.
In this paper, we consider the problem whether and how  a real-valued signal $f$  can be reconstructed,  up to a global sign,  from
 magnitude information $|f(x)|, x\in D$, or  from its phaseless samples
$|f(\gamma)|, \gamma\in \Gamma$, taken on a discrete set  $\Gamma\subset D$ in a stable way.
The above problem has been discussed for  bandlimited signals \cite{T11}
 and  wavelet signals residing in a  principal shift-invariant space \cite{YC16, CJS17, sunw17}.
 It is a  nonlinear ill-posed  problem which can be solved only if  we
have some extra information  about the signal $f$.
 In this paper,
we always assume  that the signal  $f$  has a  parametric representation,
 \begin{equation}\label{representation}
f(x)=\sum_{\lambda\in \Lambda} c_\lambda \phi_\lambda(x), \ x\in D,
\end{equation}
where
$c=(c_\lambda)_{\lambda\in \Lambda}$ is an unknown real-valued parameter vector, $\Lambda\subset D$ is a discrete set with finite  density,
and  $\Phi=(\phi_\lambda)_{\lambda\in \Lambda}$ is a  vector of nonzero basis signals
$\phi_\lambda, \lambda\in \Lambda$,  essentially supported in a neighborhood of the innovative position $\lambda\in \Lambda$.
  Those  signals  appear
    in many engineering applications such as  magnetic resonance spectrum,  ultra wide-band communication and
    electrocardiogram.
       Our representing signals of the form \eqref{representation}
     are bandlimited signals, signals in a shift-invariant space and spline signals on triangulations.
 Following the terminology in \cite{vetterli02},
      signals of the form \eqref{representation}  have finite rate of innovations  (FRI)
      and their rate of innovations is the  density of  the set $\Lambda$ \cite{donoho06, dragotti07, Sunaicm10, vetterli02}.

\smallskip
Given a signal $f$ with the parametric representation \eqref{representation},  let
 ${\mathcal M}_{f}$ contain all signals $g$ of the form   \eqref{representation} such that
 \begin{equation} \label{mf.def} |g(x)|=|f(x)|, \ x\in D.\end{equation}
  As $-f$ and $f$ have the same magnitude measurements on the whole domain, we have that  ${\mathcal M}_f\supset \{\pm f\}$.
 A natural question is  whether the above inclusion is an equality.

  \begin{question}\label{question0} Can we characterize all signals $f$ of the form \eqref{representation}
  so that ${\mathcal M}_f=\{\pm f\}$?
   \end{question}

An  equivalent statement to the above  question  is whether a  signal $f$ is determined, up to a sign, from the magnitude
information $|f(x)|, x\in D$. 
 The above question is an infinite-dimensional phase retrieval problem, which has been discussed  for bandlimited signals \cite{T11},
  wavelet signals  in a  principal shift-invariant space \cite{YC16, CJS17, sunw17}, and spatial signals in a linear space \cite{YC16}.
  The reader may refer to  \cite{Rima16, alaifari16, cahill15,   grohs17, mallat14, 
 volker14,    shenoy16}
    for historical remarks and additional references on phase retrieval in an infinite-dimensional linear space.
    In Section \ref{connectivity.section}, we introduce  an undirected graph
${\mathcal G}_f$ for a signal $f$ of the form \eqref{representation}, and
we provide an answer to  Question \ref{question0} by  showing that ${\mathcal M}_f=\{\pm f\}$ if and only if  ${\mathcal G}_f$  is connected, see Theorem \ref{pr.thm}.

\smallskip
 For a signal $f$ with a parametric representation \eqref{representation}, the  graph ${\mathcal G}_f$
is not necessarily to be connected. This leads to our next question. 

  \begin{question}\label{question1} Can we find the set ${\mathcal M}_f$ for any signal $f$ of the form \eqref{representation}?
      \end{question}
      \smallskip

 \smallskip

For a signal $f$  of the form \eqref{representation},
we can decompose its  graph ${\mathcal G}_f$
 uniquely to a union of  connected components  ${\mathcal G}_i, i\in I$, 
  \begin{equation}\label{graphdecomp.eq1}
{\mathcal G}_f= \cup_{i\in I}  {\mathcal G}_i. 
\end{equation}
Then  we can construct signals $f_i, i\in I$, of the form \eqref{representation} with ${\mathcal G}_{f_i}={\mathcal G}_i, i\in I$,
such that
\begin{equation}  \label{decomposition.def1}
 f_if_{i'}=0 \ {\rm for \ all \ distinct}\ i, i'\in  I,
\end{equation}
\begin{equation} \label{decomposition.def2} 
 {\mathcal M}_{f_i}=\{\pm f_i\}, \ i\in I,\end{equation}
and
\begin{equation} \label{decomposition.def3} f=\sum_{i\in I} f_i,
\end{equation}
see Theorem \ref{decomposition.thm00}.
Due to the mutually disjoint support property \eqref{decomposition.def1}  for  signals $f_i, i\in I$, and the connectivity
for the graphs ${\mathcal G}_{f_i}, i\in I$,
we can interprete  the above adaptive  decomposition   
 as that landscape of the original signal $f$ is composed by islands of signals $f_i, i\in I$, see Figure \ref{spline_separable.fig} and also \cite{Rima16, grohs17} for bandlimited signals.
 \begin{figure}[t] 
\begin{center}
\includegraphics[width=55mm, height=38mm]{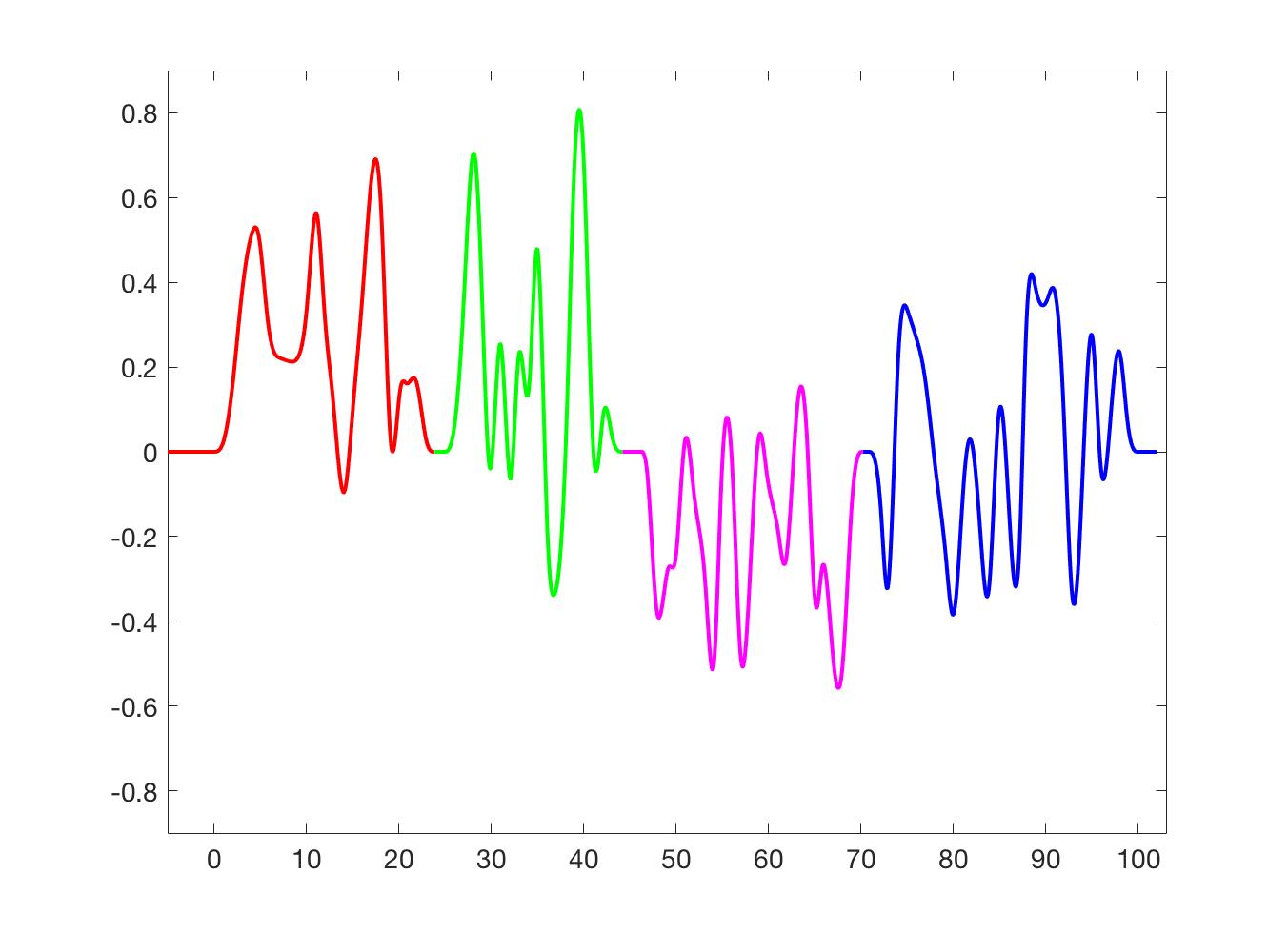}
\includegraphics[width=65mm, height=38mm]{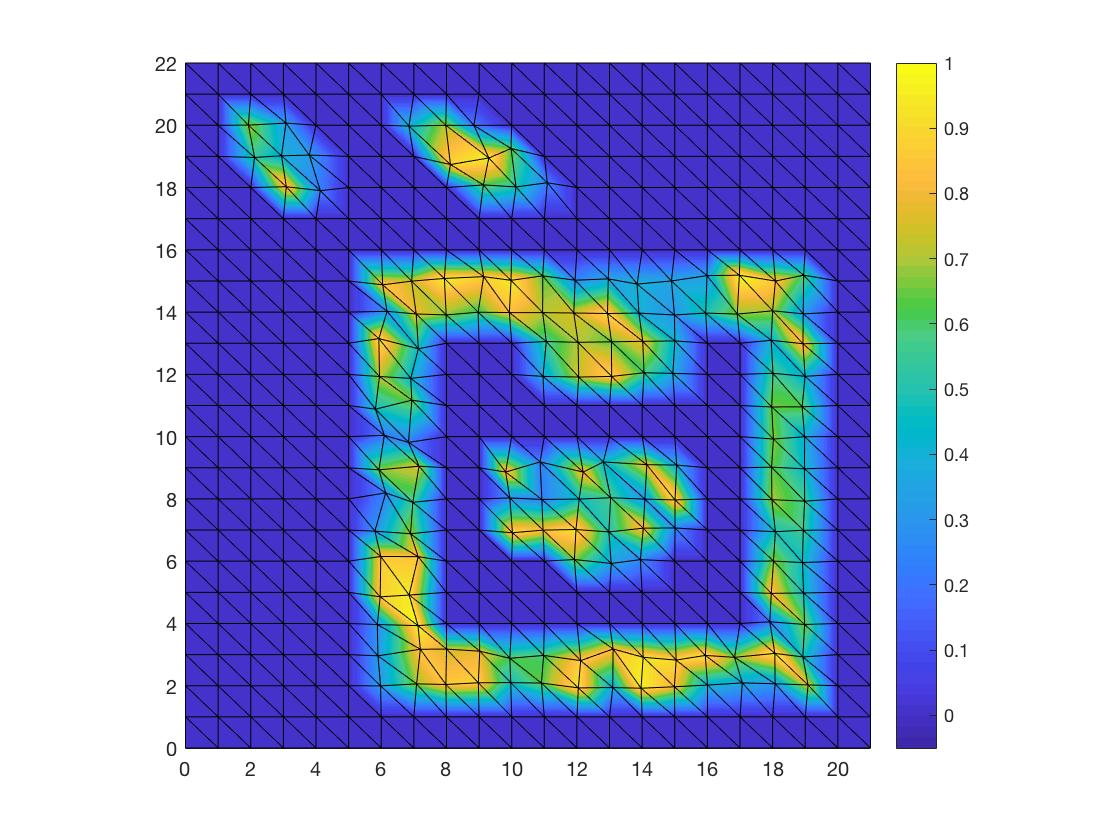}
\caption{Plotted on the left is a non-uniform cubic spline signal, while on the right is a  piecewise affine signal on a triangulation. They both have four ``islands"
 in the decomposition
 \eqref{decomposition.def1}, \eqref {decomposition.def2} and \eqref{decomposition.def3}.   }
 \label{spline_separable.fig}
\end{center}
\end{figure}

By \eqref{decomposition.def1} and \eqref{decomposition.def3}, we have
 $${\mathcal M}_f\supset \Big\{\sum_{i\in I} \delta_i f_i,\ \delta_i\in \{-1, 1\}, i\in I\Big\}.$$
 In  Section \ref{decomposition.section}, we provide an 
  answer to Question \ref{question1} by showing
 in Theorem \ref{generalphaseretrieval.thm}  that the above inclusion is in fact an equality
 for any signal $f$ of the form \eqref{representation}.
Therefore landscapes of signals $g\in {\mathcal M}_f$ are combination of
 islands of the original  signal $f$ and their  reflections.

\smallskip
Let $f$ be  a signal of the form \eqref{representation}.
To consider phaseless sampling and reconstruction on a discrete set $\Gamma\subset D$, we let
 ${\mathcal M}_{f, \Gamma}$ contain all signals $g$ of the form   \eqref{representation} such that
 \begin{equation}\label{mfgamma.def} |g(\gamma)|=|f(\gamma)|, \ \gamma\in \Gamma,\end{equation}
 and  ${\mathcal N}_\Gamma$ contain all signals $h$
of the form  \eqref{representation} such that
 \begin{equation}\label{ngamma.def} h(\gamma)=0, \ \gamma\in \Gamma.\end{equation}
 By \eqref{mf.def}, \eqref{mfgamma.def} and \eqref {ngamma.def},
 we have
 \begin{equation}\label{mfgamma.eq1}
  {\mathcal M}_f=  {\mathcal M}_{f, D},\   {\mathcal N}_D=\{0\},
 \end{equation}
 and
 \begin{equation}\label{mfgamma.eq2}
 {\mathcal M}_f +{\mathcal N}_\Gamma\subset {\mathcal M}_{f, \Gamma} \ {\rm for\ all} \ \Gamma\subset D.
 \end{equation}
This leads to the third question. 

  \begin{question}\label{question2} Can we find all discrete sets $\Gamma$ such that
   ${\mathcal M}_{f, \Gamma}= {\mathcal M}_f$
  for all signals $f$ of the form \eqref{representation}?
      \end{question}

By \eqref{mfgamma.eq2},  a necessary condition for the equality ${\mathcal M}_{f, \Gamma}= {\mathcal M}_f$ to hold for some signal $f$ of the form   \eqref{representation} is that ${\mathcal N}_\Gamma=\{0\}$, which means that all signals
of the form   \eqref{representation} are determined from their samples taken on $\Gamma$.
The reader may refer to \cite{dragotti07, sunaicm14, sunsiam06, vetterli02} and references therein for stable sampling and reconstruction of
 FRI signals. 

In Section \ref{phaseless.section}, we show the existence of  a discrete set $\Gamma$ with finite density such that ${\mathcal M}_{f, \Gamma}= {\mathcal M}_f$ for all signals  $f$ of the form \eqref{representation}.
In  Theorem \ref{samplingset2.thm}, we construct such a discrete set  $\Gamma$ explicitly    under the assumption that
the linear space for signals of the form \eqref{representation} to reside in has local complement property on a family of open sets.
The  local complement property, see Definition \ref{lcp.def0}, is introduced in \cite{CJS17} and it is closely related to
 the  complement property
  for ideal sampling functionals in \cite{YC16} and the complement property for frames in Hilbert/Banach spaces \cite{alaifari16, BCE06, Bandeira14, cahill15}.  The  local complement property on a bounded open set can be characterized by phase retrievable frames associated with
  the generator $\Phi$ and the sampling set $\Gamma$ on a finite-dimensional space, see Proposition \ref{lcpandprf.pr}.
  The reader may refer to  \cite{BBCE09, BCE06,   candes13, CSV12, casazza17,   bswx18, han2017,  jaganathany15, schechtman15,  wangxu14}
  and references therein for historical remarks and recent advances on
  finite-dimensional   phase retrievable frames.

An equivalent statement to the equality ${\mathcal M}_{f, \Gamma}={\mathcal M}_f$ is that  magnitude measurements $|f(x)|, x\in D$, on the whole domain $D$  are determined by their samples $|f(\gamma)|, \gamma\in \Gamma$, taken on a discrete set $\Gamma$.
 In practical applications,
phaseless samples are usually corrupted by some bounded deterministic/random noises $\eta(\gamma), \gamma\in \Gamma$, and the  available noisy phaseless samples are
 \begin{equation*} 
 z_\eta(\gamma)= |f(\gamma)|+\eta(\gamma),\  \gamma\in \Gamma.
 \end{equation*}
 Set  $\eta=(\eta(\gamma))_{\gamma\in \Gamma}$  and $z_\eta=(z_\eta(\gamma))_{\gamma\in \Gamma}$.
This leads to the  fourth question  to be discussed in this paper.  

\begin{question}
\label{question2}  Can we find an algorithm $\Delta$   such that the reconstructed signal $g_\eta=\Delta (z_\eta)$
is an  approximation to the original signal $f$ in  magnitude measurements?
 \end{question}

\smallskip

In Section \ref{stablereconstruction.section}, we propose an algorithm with linear complexity,
which provides an answer to Question \ref{question2}. Under the assumption that the generator $\Phi$ is well localized and uniformly bounded, we show in Theorem \ref{stability.thm} that the original signal $f$ and the reconstructed signal $g_\eta$ are well approximated by
some signals $f_\eta$ and $h_\eta$ of the form \eqref{representation} that have the same magnitude  measurements on the domain $D$. Therefore
the reconstructed signal $g_\eta$ provides  a suboptimal approximation
 to the original signal $f$ in   magnitude measurements, i.e.,  there exists an
 absolute constant  $C$ independent on the original signal $f$ and the noise $\eta$ such that
\begin{equation}\label{suboptimal.def2}
\sup_{x\in D}\big||g_\eta(x)|-|f(x)|\big |\le C \sup_{\gamma\in \Gamma} |\eta(\gamma)|.
\end{equation}
As an application of the above  estimate, we conclude that the phaseless sampling operator $S: f\longmapsto (|f(\gamma)|)_{\gamma\in \Gamma}$ is bi-Lipschitz
in magnitude measurements,
 see Corollary \ref{phaselesssamplingstability.cor}.

The paper is organized as follows.  In Section \ref{preliminaries.section},  we present some preliminaries on
the linear space $V(\Phi)$ for signals of the form \eqref{representation} to reside in.
In Section \ref{connectivity.section}, we introduce a graph structure for any signal in $V(\Phi)$ and use  its connectivity to provide an 
answer to
   Question  \ref{question0}.
 In Section \ref{decomposition.section}, we introduce a landscape decomposition for a signal $f\in V(\Phi)$ and use it to find all signals in ${\mathcal M}_f$.  In Section \ref{phaseless.section}, we construct a discrete set $\Gamma$ with finite density such that
 ${\mathcal M}_{f, \Gamma}={\mathcal M}_f$ for all $f\in V(\Phi)$.
In Section \ref{stablereconstruction.section}, we introduce a stable algorithm  $\Delta$ with linear complexity
to reconstruct signals in $V(\Phi)$ from their noisy phaseless samples taken on a discrete set $\Gamma$.
In Section \ref{simulation.section}, we demonstrate the stable reconstruction of our proposed algorithm $\Delta$
to reconstruct one-dimensional non-uniform spline signals and  two-dimensional piecewise affine signals on triangulations from their noisy phaseless samples.
In Appendix \ref{phaseless.appendix},
we show  that the  density of a discrete set $\Gamma$ with ${\mathcal M}_{f, \Gamma}={\mathcal M}_f, f\in V(\Phi)$,
 must be no less than the innovation rate of signals in $V(\Phi)$.

\section{Preliminaries} 
\label{preliminaries.section}  

In this section,  we present some preliminaries on the domain $D$ for signals to define and
the linear space $V(\Phi)$ for signals with the parametric expression \eqref{representation} to reside in. 
%

\smallskip

Spatial signals  in this paper are defined on a domain $D$.
  Our representing models
 are the $d$-dimensional  Euclidean space $\Rd$, the $d$-dimensional torus $\T^d$ and  the  simple graph to describe a spatially distributed network \cite{CJSacha18}. In this paper, we always assume the following for the domain $D$ \cite{CJSacha18,  macias1979, Yangbook2013}.

\begin{assumption}  \label{domain.assumption}
  The domain $D$ is equipped with a distance $\rho$ and a Borel measure $\mu$  so that
\begin{equation} \label {domain.assumption.eq2}
\sup_{x\in D} \mu\big(B(x, r)\big)<\infty
\end{equation}
and
 \begin{equation}\label{density.thm.eq1}
\liminf_{s\to \infty} \inf_{x\in D} \frac{\mu(B(x, s-r))}{\mu(B(x, s))}=1,\ r\ge 0,
\end{equation}
 where $B(x, r)=\{y\in D: \ \rho(x,y)\le r\}$ is the closed ball with center $x$ and  radius $r$.
\end{assumption}

\smallskip


 Spatial signals  $f$  with the  parametric representation
\eqref{representation}
 reside in the linear space
\begin{equation} \label{Vphi.def}
V(\Phi):=\Big\{\sum_{\lambda\in \Lambda}  c_\lambda \phi_\lambda: \ c_\lambda\in \R\ {\rm for \ all}\ \lambda\in \Lambda\Big\}\end{equation}
generated by $\Phi=(\phi_\lambda)_{\lambda\in \Lambda}$. 
Denote  the cardinality of a set $E$ by $\# E$.
In this paper, we always assume the following to basis signals $\phi_\lambda, \lambda\in \Lambda$.

\begin{assumption}
\label{generator.assumption}
The discrete set $\Lambda$
 has finite density
  \begin{equation}\label{generator.assumption.eq1} 
D_+(\Lambda):=\limsup\limits_{{r\to \infty}}\ \sup_{x\in D} \frac{\sharp\big(\Lambda\cap B(x, r)\big)}{\mu\big(B(x, r)\big)}<\infty,\end{equation}
the nonzero basis signals $\phi_\lambda, \lambda\in \Lambda$, are  continuous and supported in   balls with center $\lambda$ and  fixed radius $r_0>0$ independent of $\lambda$,
\begin{equation} \label{generator.assumption.eq2}\phi_\lambda (x)=0 \ {\rm for \ all} \ x\not\in B(\lambda, r_0),\  \lambda\in \Lambda;
\end{equation}
 and any signal in $V(\Phi)$ has
 a  unique parametric representation \eqref{representation}.

\end{assumption}

The prototypical forms of the space $V(\Phi)$  
are principal shift-invariant spaces generated by
the shifts of a compactly supported function $\phi$,
 twisted shift-invariant spaces generated by (non-)uniform Gabor frame system (or Wilson basis) in the time-frequency analysis (see \cite{BCHL06, pete, G_book, Janssen95, Ron97} and references  therein), and nonuniform spline signals \cite{BTU04, HA78, Schumaker}.
  The linear space $V(\Phi)$ was introduced in \cite{Sun08, sunsiam06} to model signals with finite rate of innovation (FRI).
Following the terminology in \cite{vetterli02},  signals in the linear space $V(\Phi)$ have rate   of innovation  $D_+(\Lambda)$
and
 innovative positions  $\lambda\in \Lambda$.

An equivalent statement to the  unique parametric representation \eqref{representation}
of signals in $V(\Phi)$ is that
 the generator $\Phi$ has {\em global linear independence}, i.e.,
  the map
\begin{equation} \label{generator.assumption.eq3}
  c:= (c_\lambda)_{\lambda\in \Lambda}\longmapsto   c^T\Phi:=\sum_{\lambda\in \Lambda} c_\lambda \phi_\lambda\end{equation}
  is one-to-one  from the  space $\ell(\Lambda)$ of all sequences on $\Lambda$ to the linear space  $V(\Phi)$ \cite{jia92, ron89}.
For any open set $A$, define
\begin{equation}\label{KA.def}
K_A=\{\lambda\in \Lambda:  \  \phi_\lambda \not\equiv  0\ {\rm  on} \ A\}.
 \end{equation}
 A strong version  of  the global linear independence  \eqref{generator.assumption.eq3}
is {\em local linear independence} 
 on an open set $A\subset D$, i.e.,
 \begin{equation}\label{locallinearindependence.def}  \dim V(\Phi)|_A= \# K_A,\end{equation}
 where  for a linear space $V$ we denote its dimension and  restriction on  a set $A$ by $\dim V$ and $V|_A$
  respectively.
 Observe that
 the restriction of  the linear space $V(\Phi)$ on an bounded open set $A$
is generated by $\phi_\lambda, \lambda\in K_A$. Then an equivalent formulation of  the local linear independence on a bounded open set $A$
 is that
 \begin{equation}\label{locallinearindependence.def}
\sum_{\lambda\in \Lambda} c_\lambda\phi_\lambda(x) =0 \  {\rm  on } \ x\in A\end{equation} implies that  $ c_\lambda=0$ for all $\lambda\in K_A$  \cite{jia92, Sunaicm10}.

Set
\begin{equation}\label{graph.remark.eq3}
S_\Phi(\lambda, \lambda'):=\{x\in D: \ \phi_\lambda (x) \phi_{\lambda'}(x) \ne 0 \}, \  \lambda, \lambda'\in \Lambda,
\end{equation}
and use the abbreviation $S_\Phi(\lambda):=S_{\Phi}(\lambda, \lambda)$ when $\lambda'=\lambda\in \Lambda$.
 One may verify that the generator  $\Phi$ has global linear independence  \eqref{generator.assumption.eq3} if it has local linear independence on  a family of open sets
$T_\theta, \theta\in \Theta$, such that
\begin{equation}
\label{Ttheta.def}
 S_\Phi(\lambda, \lambda')
\cap \big(\cup_{\theta\in \Theta} T_\theta\big)\ne \emptyset
\end{equation}
 for  all pairs $(\lambda, \lambda')\in \Lambda\times \Lambda$ with $S_\Phi(\lambda, \lambda')\ne \emptyset$.
We remark that  a family of open sets
$T_\theta, \theta\in \Theta$,  satisfying \eqref{Ttheta.def}
is not necessarily a covering of the domain $D$, however,
the converse is true, cf. Corollary \ref{decomposition.cor}.

\section{Phase retrievability and graph connectivity} 
\label{connectivity.section}

In this section, we characterize all signals $f\in V(\Phi)$
 that are determined, up to a sign, from their magnitude measurements 
on the  whole domain $D$, i.e.,  ${\mathcal M}_f=\{\pm f\}$.

Given  a signal $f=\sum_{\lambda\in \Lambda} c_\lambda \phi_\lambda\in V(\Phi)$,
we  define an
 undirected graph
\begin{equation}\label{graphG.def}
{\mathcal G}_f:=(V_f, E_f),\end{equation}
where
 \begin{equation}
 V_f:=\{\lambda\in \Lambda: \ c_{\lambda}\neq  0\}\end{equation}
and
 \begin{equation*} E_f := \big\{({\lambda},{\lambda}')\in  V_f\times V_f:  \ { \lambda}\neq {\lambda}' \ {\rm and} \
\phi_{\lambda}\phi_{\lambda'} \not\equiv 0 \big\}.
\end{equation*}
 For a signal $f\in V(\Phi)$, the graph ${\mathcal G}_f$  in \eqref{graphG.def} is well-defined  by \eqref{generator.assumption.eq3},
 and it   was introduced in \cite{CJS17}
  when the generator $\Phi=(\phi(\cdot-k))_{k\in \Z^d}$ is obtained from  shifts of a compactly supported function $\phi$.
    Its vertex set $ V_f$
  contains all innovative positions $\lambda\in \Lambda$ with nonzero amplitude  $c_\lambda$,   and
its edge set $ E_f$
contains all innovative position pairs $(\lambda, \lambda')$  in $V_f\times V_f$ with
basis signals $\phi_\lambda$ and $\phi_{\lambda'}$ having
 overlapped supports, i.e.,  
\begin{equation}\label{graph.remark.eq1}
(\lambda ,\lambda')\in E_f \ {\rm if\ and\ only\ if} \ \lambda, \lambda'\in V_f \ {\rm and} \
(\lambda, \lambda')\in E_\Phi,
\end{equation}
 where   $S_\Phi(\lambda, \lambda'), (\lambda, \lambda')\in \Lambda\times \Lambda$, are given in
 \eqref{graph.remark.eq3} and
 \begin{equation}  \label{graph.remark.eq2}
 E_\Phi:=\{(\lambda, \lambda')\in \Lambda\times \Lambda: \ S_\Phi(\lambda, \lambda')\ne \emptyset\}.
 \end{equation}

  To study phase retrievability  of signals in $V(\Phi)$,
 we recall the local complement property  of  a linear space  of  real-valued  signals
 \cite{CJS17}.

\begin{definition}\label{lcp.def0} 
\rm Let   $A$ be an open subset of the domain $D$.
We say that a linear space $V$  of real-valued signals on the domain $D$  has {\em local complement property} on  $A$ if
for any $ A'\subset  A$, there does not exist  $f, g\in V$ such that $f, g\not\equiv 0$ on $A$, but $f({ x})=0$ for all ${x}\in  A'$ and $g({y})=0$ for all ${ y}\in  A\backslash  A'$.
 \end{definition}

The  local complement property  is the  complement property  in \cite{YC16}
  for ideal sampling functionals on a set, cf. the complement property for frames in Hilbert/Banach spaces (\cite{alaifari16, BCE06, Bandeira14, cahill15}). Local complement property is closely related to local phase retrievability.
In fact, following the argument in \cite 
{YC16}, the  linear space $V$ has the local complement property on $A$ if and only if all signals in $V$ is {\em local phase retrievable} on $A$, i.e., for any $ f, g\in V$ satisfying $|g({x})|=|f({x})|, { x}\in A$, there exists $\delta\in \{-1, 1\}$ such that $g({ x})=\delta f({x})$ for all ${x}\in  A$.

\smallskip

%
%
%
In this section,
 we establish the equivalence between
 phase retrievability of a nonzero signal $f\in V(\Phi)$ and
 connectivity  of  its graph ${\mathcal G}_f$, which is  established in \cite{CJS17} for signals residing in a principal shift-invariant space.

\begin{theorem}\label{pr.thm}
Let  $D$ be a  domain satisfying Assumption \ref{domain.assumption},
  $\Phi: = (\phi_\lambda)_{\lambda\in \Lambda}$  
be  a family of basis functions  satisfying  Assumption \ref{generator.assumption},
${\mathcal T}:=\{T_\theta, \theta\in \Theta\}$ be a family of open sets satisfying
\eqref{Ttheta.def},
and let $V(\Phi)$  be the linear space \eqref{Vphi.def}  generated by $\Phi$.
Assume that for any $T_\theta\in \mathcal T$,   $\Phi$  
has local linear independence on  $T_\theta$ and  $V(\Phi)$ has local complement property on $T_\theta$.
Then for a nonzero signal $f\in V(\Phi)$, ${\mathcal M}_f=\{\pm f\}$ 
 if and only if
 the   graph ${\mathcal G}_f$ in \eqref{graphG.def} is connected.
\end{theorem}

We  remark that the local complement assumption in Theorem \ref{pr.thm} is satisfied when
$\Phi$ has local linear independence on all open sets.

\begin{proposition}\label{llcimplieslc.pr}
 Let  $\Phi: = (\phi_\lambda)_{\lambda\in \Lambda}$  
satisfy Assumption \ref{generator.assumption}.
If  $\Phi$ has local linear independence on all open sets, then
 there exist ${\mathcal T}:=\{T_\theta, \theta\in \Theta\}$ satisfying  \eqref{Ttheta.def} 
such that $V(\Phi)$ has local complement property on every $T_\theta\in {\mathcal T}$. 
\end{proposition}

\begin{proof}   
Define
$S_{\Phi}(\theta)=\cap_{\lambda\in \theta} S_{\Phi}(\lambda)$ for  a set $\theta\subset \Lambda$. We say that $\theta\subset \Lambda$ is maximal if  $S_{\Phi}(\theta)\ne \emptyset$ and $S_{\Phi}(\theta')=\emptyset$ for all $\theta'\supsetneq 
\theta$.
By \eqref{generator.assumption.eq1} and \eqref{generator.assumption.eq2}, any maximal set  contains finitely many elements.
Denote the family of all maximal sets  by $\Theta$ and define
$T_\theta= S_\Phi(\theta), \theta \in \Theta$.  Clearly ${\mathcal T}:=\{T_\theta, \theta\in \Theta\}$ satisfies
\eqref{Ttheta.def}, because any $\theta\subset \Lambda$ with $S_\Phi(\theta)\ne \emptyset$ is a subset of some maximal set in  $\Theta$.

Now it  remains to prove that $V(\Phi)$ has local complement property on $T_\theta, \theta \in \Theta$. Take an arbitrary $\theta \in \Theta$
and two signals $f,g\in V(\Phi)$ satisfying $|f(x)|=|g(x)|$  for all $x\in T_\theta$.
Then
\begin{equation}\label{llcimplieslc.pr.pf.eq1}
(f+g)(x) (f-g)(x)=0\ {\rm  for\ all} \ x\in T_\theta.
\end{equation} Write $f+g=\sum_{\lambda\in \Lambda} c_\lambda\phi_{\lambda}$ and
$f-g=\sum_{\lambda\in \Lambda} d_\lambda\phi_{\lambda}$, and  set $B_1=\{x\in T_\theta: \  (f+g)(x)\ne 0\}$ and $B_2=\{x\in T_\theta: \ (f-g)(x)\ne 0\}$.
Then
\begin{equation}\label{llcimplieslc.pr.pf.eq2}
\Big(\sum_{\lambda\in \theta} c_\lambda\phi_{\lambda}(x)\Big)  
\Big( \sum_{\lambda\in \theta} d_\lambda\phi_{\lambda}(x)
\Big)=0\ {\rm  for\ all} \ x\in T_\theta,
\end{equation}
and
\begin{equation}\label{llcimplieslc.pr.pf.eq3}
 \ \phi_{\lambda}(x)\ne 0 \ {\rm for \ all}  \ x\in T_\theta \ {\rm and } \ \lambda\in \theta
\end{equation}
by assumption \eqref{Ttheta.def}, \eqref{llcimplieslc.pr.pf.eq1} and the construction of maximal sets.
By \eqref{llcimplieslc.pr.pf.eq2}, we have that either $f-g=0$ on $B_1$, or $f+g=0$ on $B_2$, or $f-g=f+g=0$ on $T_\theta$.
This together with \eqref{llcimplieslc.pr.pf.eq3} and the local linear independence on $B_1$ or $ B_2$ or $T_\theta$ implies that
either   $d_\lambda=0$  for all $\lambda \in \theta$, or $c_\lambda=0$ for all $\lambda\in \theta$, or $c_\lambda=d_\lambda=0$ for all $\lambda\in \theta$.
  Therefore either
$f=g$ on $T_\theta$, or $f=-g$ on $T_\theta$, or $f=g=0$ on $T_\theta$. This completes the proof.
\end{proof}

Applying  Theorem \ref{pr.thm} and Proposition \ref{llcimplieslc.pr}, we have the following  corollary, which is established
in \cite{CJS17} when the generator $\Phi$ is obtained from shifts of a compactly supported function.

\begin{corollary}\label{pr.cor} Let  $D$ be a domain satisfying Assumption \ref{domain.assumption},
  $\Phi$  
be  a family of basis functions  satisfying  Assumption \ref{generator.assumption},
and let $V(\Phi)$  be the linear space \eqref{Vphi.def}  generated by $\Phi$.
If    $\Phi$  
has local linear independence on any open sets, then
 a nonzero signal $f\in V(\Phi)$  satisfies ${\mathcal M}_f=\{\pm f\}$ if and only if
 the   graph ${\mathcal G}_f$ in \eqref{graphG.def} is connected.
\end{corollary}

\subsection{Proof of Theorem  \ref{pr.thm}}
\label{pr.thm.pf.section}
The necessity in Theorem \ref{pr.thm}
 holds under a weak assumption on the generator $\Phi$.

 \begin{proposition}\label{necessary.pr}
 Let  $D$ be a domain satisfying Assumption \ref{domain.assumption},
  $\Phi: = (\phi_\lambda)_{\lambda\in \Lambda}$  
be  a family of basis functions  satisfying  Assumption \ref{generator.assumption},
 $V(\Phi)$  be the linear space \eqref{Vphi.def}  generated by $\Phi$, and let $f$ be a nonzero signal in $V(\Phi)$.
If  ${\mathcal M}_f=\{\pm f\}$, then the   graph ${\mathcal G}_f$ in \eqref{graphG.def} is connected.
 \end{proposition}

To prove Proposition \ref{necessary.pr}, we recall a characterization in \cite{YC16} on phase retrievability.

\begin{lemma}\label{nonseparable.lem}  For a nonzero signal $f$ in a linear space $V$, ${\mathcal M}_f=\{\pm f\}$
 if and only if it is nonseparable, i.e.,
 there
does not exist nonzero signals $f_0$ and $f_1\in V$ such that
\begin{equation}\label{separable.def} f=f_0+f_1 \ \ {\rm and}\  \ f_0f_1=0.
\end{equation}
\end{lemma}

%
\begin{proof}[Proof of Proposition \ref{necessary.pr}]Let $f\in V(\Phi)$ be a nonzero signal satisfying ${\mathcal M}_f=\{\pm f\}$,
and write $f=\sum_{\lambda \in \Lambda} c_{\lambda}\phi_\lambda$, where $c_{\lambda}\in \R, \lambda\in \Lambda$.
Suppose, on the contrary, that the graph ${\mathcal G}_f$ is disconnected.  Then  there exists  a  nontrivial connected component  $W$ such that both $W$ and $V_f\backslash W$ are nontrivial, and   
 no edges exist between vertices in $W$ and in $V_f\backslash W$.
Write
 \begin{equation}\label{necessarypr.thm.pf.eq1}
 f= 
\sum_{k\in V_f} c_{\lambda}\phi_\lambda
 = \sum_{{ \lambda}\in W}c_{\lambda}\phi_\lambda+\sum_{{ \lambda}\in V_f\backslash W}c_{\lambda}\phi_\lambda 
 =: f_0+f_1.\end{equation}
From the global linear independence \eqref{generator.assumption.eq3} and nontriviality of the sets $W$  and $V_f\backslash W$, we obtain
\begin{equation} \label{necessarypr.thm.pf.eq2}
f_0\not\equiv 0 \ \ {\rm  and}\  \  f_1\not\equiv 0.\end{equation}
Applying \eqref{necessarypr.thm.pf.eq1} and \eqref {necessarypr.thm.pf.eq2},  and using the  characterization
in Lemma
\ref{nonseparable.lem},  we obtain that $$f_0(x_0)f_1(x_0)\neq 0$$ for some $x_0\in D$. This implies
the existence of ${\lambda}\in W$ and  ${\lambda}'\in  V_f\backslash W$ such that
 $c_\lambda \phi_\lambda(x_0)\neq 0$ and  $c_{\lambda'}\phi_{\lambda'}(x_0)\neq 0$.
   Hence $({\lambda}, {\lambda}')$ is an edge between $ \lambda\in W$ and ${ \lambda}'\in V_f\backslash W$,  which contradicts to the construction of the set $W$.  
\end{proof}
   \smallskip

Now we prove the sufficiency in Theorem \ref{pr.thm}. Let $f=\sum_{\lambda\in \Lambda} c_\lambda \phi_\lambda\in V(\Phi)$ have its graph ${\mathcal G}_f$ being connected, and take $g=\sum_{\lambda\in \Lambda} d_\lambda \phi_\lambda\in {\mathcal M}_f$. Then
for any $\theta\in \Theta$,
\begin{equation}\label{pr.thm.pf.eq20-}
|g(x)|=|f(x)|,\ x\in T_\theta.
\end{equation}
 For any $\theta\in \Theta$, there exists
$\delta_\theta \in \{-1, 1\}$ by  \eqref{pr.thm.pf.eq20-} and the local complement property
on $T_\theta$ such that
\begin{equation*}
g(x)=\delta_\theta f(x),\  x\in T_\theta.
\end{equation*}
This together with the local linear independence on $T_\theta$ implies that
\begin{equation}\label{pr.thm.pf.eq20} d_\lambda=\delta_\theta c_\lambda
\end{equation}
for all $\lambda\in \Lambda$ with $S_\Phi(\lambda)\cap T_\theta\ne \emptyset$.
 Using  \eqref{Ttheta.def} 
and applying \eqref{pr.thm.pf.eq20}, there exist $\delta_\lambda\in \{-1, 1\}, \lambda\in \Lambda$ such that
  \begin{equation} \label{pr.thm.pf.eq21} d_\lambda=\delta_\lambda c_\lambda
  \end{equation}
for all $\lambda\in \Lambda$,  and
  \begin{equation} \label{pr.thm.pf.eq22} \delta_\lambda=\delta_{\lambda'}
  \end{equation}
  for any edge $(\lambda, \lambda')$ in the graph ${\mathcal G}_f$.  Combining
  \eqref{pr.thm.pf.eq21} and \eqref{pr.thm.pf.eq22}, and  applying connectivity of the graph ${\mathcal G}_f$, we can find  $\delta\in \{-1, 1\}$ such that
  \begin{equation}  d_\lambda=c_\lambda=0 \ {\rm for \ all} \ \lambda\not\in V_f \ {\rm and} \ d_\lambda=\delta c_\lambda \ {\rm
  for \ all} \ \lambda\in V_f.
  \end{equation}
  Thus
  $g(x)=\delta f(x)$ for all $ x\in D$.
   This completes the proof of the sufficiency.

\section{Phase nonretrievability and landscape decomposition}
\label{decomposition.section}

Given a signal $f\in V(\Phi)$, the  graph ${\mathcal G}_f$ in \eqref{graphG.def}
is not necessarily to be connected and hence
  there may exist   signals  $g\in V(\Phi)$, other than $\pm f$,   belonging to ${\mathcal M}_f$. 
 In this section, we characterize the  set ${\mathcal M}_f$ of all signals $g\in V(\Phi)$ that have the same magnitude measurements  on  the domin $D$ as $f$ has, and then we provide the answer to  Question \ref{question1}.

Take  $f=\sum_{\lambda \in \Lambda}
c_\lambda \phi_\lambda\in  V(\Phi)$,
let ${\mathcal G}_i=(V_i, E_i), i\in I$, be connected components of  the  graph ${\mathcal G}_f$, 
and define
 \begin{equation}\label{graphdecomp.eq2} f_i= \sum_{\lambda \in V_i} c_\lambda \phi_\lambda,\ i\in I.
\end{equation}
Then  \eqref{graphdecomp.eq1} holds by the definition of ${\mathcal G}_i, i\in I$, and  the signal $f$ has the decomposition \eqref{decomposition.def1}, \eqref{decomposition.def2} and \eqref{decomposition.def3}  
   by Theorem \ref{pr.thm}. 
By \eqref{decomposition.def1} and \eqref{decomposition.def3}, signals $g=\sum_{i\in I} \delta_i f_i$ with $\delta_i\in \{-1, 1\}, i\in I$, have the same magnitude measurements  on  the domain $D$ as $f$ has. In the following theorem, we show that the converse is also true.

 \begin{theorem}\label{generalphaseretrieval.thm}
  Let  the domain $D$, the generator
$\Phi: = (\phi_\lambda)_{\lambda\in \Lambda}$,  
the family ${\mathcal T}:=\{T_\theta, \theta\in \Theta\}$  of open sets,
  and the linear space $V(\Phi)$  be as in Theorem \ref{pr.thm}.
Take  $f\in V(\Phi)$ and let $f_i\in V(\Phi), i\in I$, be as in 
\eqref{graphdecomp.eq2}.
 Then  $g\in V(\Phi)$  belongs to ${\mathcal M}_f$
 if and only if
 \begin{equation}\label{generalphaseretrieval.thm.eq1}
 g=\sum_{i\in I} \delta_i f_i \ \  {\rm for \ some }   \ \delta_i\in \{-1, 1\}, i\in I.
 \end{equation}   \end{theorem}

The conclusion in  Theorem \ref{generalphaseretrieval.thm} can be 
understood as that  the landscape of  any signal $g\in {\mathcal M}_f$ is  a combination of
 islands of the original signal $f$ or their  reflections.  As an application to Theorem \ref{generalphaseretrieval.thm}, we have the following result about the cardinality of the set ${\mathcal M}_f$.

\begin{corollary} Let  the domain $D$, the generator
$\Phi$,  
the family ${\mathcal T}$  of open sets
  and the linear space $V(\Phi)$  be as in Theorem \ref{pr.thm}.
  Then for $f\in V(\Phi)$,
  $$\# {\mathcal M}_f= 2^{\# I},$$
  where $I$ is given in \eqref{graphdecomp.eq1}.
\end{corollary}

   To prove 
   Theorem \ref{generalphaseretrieval.thm}, we need the uniqueness
   of a landscape decompositions satisfying  \eqref{decomposition.def1}, \eqref{decomposition.def2}
   and \eqref{decomposition.def3}.

 \begin{theorem}\label{decomposition.thm00}
 Let the generator $\Phi$ and the  space  $V(\Phi)$ be as in Theorem \ref{generalphaseretrieval.thm}.
 Then for any $f\in V(\Phi)$ there exists a unique decomposition
 satisfying \eqref{decomposition.def1}, \eqref{decomposition.def2} and \eqref{decomposition.def3}.
  \end{theorem}

\begin{proof} 
  Write $f=\sum_{\lambda \in \Lambda} c_\lambda \phi_\lambda$.
  First we prove the existence of a decomposition   satisfying \eqref{decomposition.def1}, \eqref{decomposition.def2} and \eqref{decomposition.def3}.
 Define $f_i, i\in I$, as in
 \eqref{graphdecomp.eq2}. 
  Then the decomposition \eqref{decomposition.def3} holds by \eqref{graphdecomp.eq1} and \eqref{graphdecomp.eq2}, and the nonseparability property
  \eqref{decomposition.def2} of $f_i, i\in I$  follows from Theorem \ref{pr.thm} 
   and the connectivity of ${\mathcal G}_i, i\in I$.
  Recall that  there are no edges between vertices in different connected components ${\mathcal G}_i, i\in I$.
This leads to the mutually disjoint support property \eqref{decomposition.def1}.

Now we prove the uniqueness of the decomposition
  \eqref{decomposition.def1}, \eqref{decomposition.def2} and \eqref{decomposition.def3}.
Let $0\ne g_j=\sum_{\lambda \in \Lambda} d_{j, \lambda} \phi_\lambda\in V(\Phi), j\in J,$
satisfy
 \begin{equation}\label{decomposition.thm.pf.eq1}
 f=\sum_{j\in J} g_j,
 \end{equation}
 \begin{equation} \label{decomposition.thm.pf.eq3}
{\mathcal M}_{g_j}=\{\pm g_j\},\ j\in J,  
 \end{equation}
 and
 \begin{equation} \label{decomposition.thm.pf.eq2}
 g_jg_{j'}=0 \ {\rm for \ all\ distinct} \ j, j'\in J.
 \end{equation}
Then it suffices to find  a  partition $I_j, j\in J$, of the set $I$
  such that
  \begin{equation}  \label{decomposition.thm.pf.eq3+}
g_j=\sum_{i\in I_j } f_i\end{equation}
where $f_i, i\in I$, are given in   \eqref{graphdecomp.eq2}, and that \begin{equation}  \label{decomposition.thm.pf.eq3++}
I_j {\rm \ only \  contains \  exactly \  one \  element \ for\  any}  \   j \in J.\end{equation}
First we  prove \eqref{decomposition.thm.pf.eq3+}.
For any distinct $j, j'\in J$ and  $(\lambda, \lambda')\in \Lambda\times \Lambda$ with
 $S_\Phi(\lambda, \lambda')\ne\emptyset$,  following the argument used in the sufficiency of Theorem
 \ref{pr.thm} with $f$ and $g$ replaced by  $g_j\pm g_{j'}$ we obtain from \eqref{decomposition.thm.pf.eq2} that
   \begin{equation*} \label{decomposition.thm.pf.eq5*}
 {\rm either} \ (d_{j,\lambda}, d_{j, \lambda'})=(0, 0)\ {\rm  or}\  (d_{j',\lambda}, d_{j', \lambda'})=(0, 0).
 \end{equation*}
 This together with  \eqref{decomposition.thm.pf.eq1} 
 implies the existence of $j\in J$ such that
 \begin{equation}\label{decomposition.thm.pf.eq4}
 d_{j, \lambda}= c_\lambda, \ d_{j, \lambda'} = c_{\lambda'}
 \end{equation}
 and
 \begin{equation} \label{decomposition.thm.pf.eq5}
 d_{j',\lambda}=  d_{j', \lambda'}= 0\ {\rm for\ all}\ j'\ne j.
 \end{equation}
Observe that $S_\Phi(\lambda)\ne \emptyset, \lambda\in \Lambda$. Applying \eqref{decomposition.thm.pf.eq4} and \eqref{decomposition.thm.pf.eq5}  with $\lambda'= \lambda\in \Lambda$, we can find a mutually disjoint partition $W_j, j\in J$, of the  set $V_f$  such that
  \begin{equation}  \label{decomposition.thm.pf.eq5+}
 g_j=\sum_{\lambda\in W_j} c_\lambda\phi_\lambda.
 \end{equation}
Applying \eqref{decomposition.thm.pf.eq4} and \eqref{decomposition.thm.pf.eq5} with $(\lambda, \lambda')$ being an edge in ${\mathcal G}_f$, 
we obtain that for any $i\in I$ there exists $j\in J$ such that $V_i\subset W_j$. This together with
\eqref{graphdecomp.eq2},
\eqref{decomposition.thm.pf.eq5+} and the observation $\cup_{i\in I} V_i=\cup_{j\in J} W_j=V_f$
 proves  \eqref{decomposition.thm.pf.eq3+}.

  Now we prove  \eqref{decomposition.thm.pf.eq3++}. By  \eqref{decomposition.def1}
  and  \eqref{decomposition.thm.pf.eq3+} we have that
  $${\mathcal M}_{g_j}\supset \Big\{\sum_{i\in  I_j} \delta_i f_i, \delta_i\in \{-1, 1\}\Big\},$$
which implies that
$\#  {\mathcal M}_{g_j}\ge 2^{\# I_j}$. This together with \eqref{decomposition.thm.pf.eq3} proves
\eqref{decomposition.thm.pf.eq3++}.
  \end{proof}

Now we start to prove Theorem \ref{generalphaseretrieval.thm}.
\begin{proof}[Proof of Theorem \ref{generalphaseretrieval.thm}]  The sufficiency is obvious. Now the necessity. Let $f, g\in V(\Phi)$ have the same magnitude measurements on the domain $D$, i.e., ${\mathcal M}_f={\mathcal M}_g$. Write
$f=\sum_{\lambda\in \Lambda} c_\lambda \phi_\lambda$ and $g=\sum_{\lambda\in \Lambda} d_\lambda \phi_\lambda$.
Then following the argument used in the sufficiency of Theorem \ref{pr.thm},
 we can find $\delta_{\lambda, \lambda'}\in \{-1, 1\}$ for any pair $(\lambda, \lambda')$ with  $S_\Phi({\lambda, \lambda'})\ne \emptyset$
 such that
 \begin{equation}  \label{generalphaseretrieval.thm.pf.eq2}
 (d_{\lambda}, d_{\lambda'})= \delta_{\lambda, \lambda'}   (c_\lambda, c_{\lambda'}).
 \end{equation}
Applying \eqref{generalphaseretrieval.thm.pf.eq2} with $\lambda'=\lambda$ and recalling that $S_\Phi(\lambda)\ne \emptyset$, we obtain
\begin{equation} \label {generalphaseretrieval.thm.pf.eq3}
d_\lambda= \delta_\lambda c_\lambda, \ \lambda \in \Lambda,
\end{equation}
for some $\delta_\lambda\in \{-1, 1\}$.
This concludes that
\begin{equation}
\delta_\lambda=\delta_{\lambda, \lambda'}=\delta_{\lambda'}
\end{equation}
for any edge $(\lambda, \lambda')$ of the graph ${\mathcal G}_f$.
Therefore the signs $\delta_\lambda$ are the same in any connected component of the graph ${\mathcal G}_f$. This together with  \eqref{graphdecomp.eq1}, \eqref{graphdecomp.eq2} and \eqref{generalphaseretrieval.thm.pf.eq3} completes the proof.
\end{proof}

The union of $T_\theta, \theta\in \Theta$, is not necessarily the whole domain $D$. Following the argument used in the proof of Theorems \ref{pr.thm} and \ref{generalphaseretrieval.thm},
 we have the following corollary.

\begin{corollary}\label{decomposition.cor}
  Let  the domain $D$, the generator
$\Phi$,  
the family of open sets  ${\mathcal T}=\{T_\theta, \theta\in \Theta\}$ and the linear space $V(\Phi)$  be as in Theorem \ref{generalphaseretrieval.thm}.
Then
\begin{equation} {\mathcal M}_f={\mathcal M}_{f, D_{\mathcal T}} \ {\rm for \ all}\ f\in V(\Phi),
\end{equation}
where $D_{\mathcal T}=\cup_{\theta\in \Theta} T_\theta$.
\end{corollary}

\begin{proof} Let $f, g\in V(\Phi)$ satisfy $|f(x)|=|g(x)|, x\in T_\theta$ for all $\theta\in \Theta$.
Write $f=\sum_{i\in I} f_i$ as in  \eqref{decomposition.def1}, \eqref{decomposition.def2} and  \eqref{decomposition.def3}. From the argument used in the proof of Theorems  \ref{pr.thm} and \ref{generalphaseretrieval.thm}, we have
that $g=\sum_{i\in I} \delta_i f_i$ for some $\delta_i\in \{-1, 1\}$. Therefore $|g(x)|=|f(x)|$ for all $x\in D$.
\end{proof}

\section{Phaseless sampling and reconstruction }
\label{phaseless.section}

  To study  phaseless sampling and reconstruction of signals in $V(\Phi)$,
we recall the concept of a phase retrievable frame \cite{BCE06,  bswx18, han2017, wangxu14}.

 \begin{definition} \rm
 We say that ${\mathcal F}=\{f_m\in \R^n, 1\le m\le M\}$ is
 a {\em phase retrievable frame} for $\R^n$ if
any vector $v\in \R^d$ is determined, up to a sign, by its measurements
$|\langle v, f_m\rangle|, f_m\in {\mathcal F}$, and that ${\mathcal F}$ is a
{\em minimal phase retrieval frame} for $\R^n$
if any true subset of ${\mathcal F}$ is not a phase retrievable frame.
\end{definition}

It is known that a  minimal phase retrieval frame  for $\R^n$ contains at least $2n-1$ vectors and at most $n(n+1)/2$  vectors
\cite{BCE06,  CJS17, han2017}.
In this section, we construct a discrete set  $\Gamma$  with finite density such that
\begin{equation}\label{samplingproblem} {\mathcal M}_{f, \Gamma}={\mathcal M}_f\ {\rm for\ all} \  f\in V(\Phi).
\end{equation}

 \begin{theorem}\label{samplingset2.thm}   Let  the domain $D$, the generator
$\Phi: = (\phi_\lambda)_{\lambda\in \Lambda}$,
the family ${\mathcal T}=\{T_\theta, \theta\in \Theta\}$ of open sets, 
and the linear space $V(\Phi)$  be as in Theorem \ref{pr.thm}.
 Set \begin{equation}\label{rs.def}
R_\Lambda(r):=\sup_{x\in D} \#\big(\Lambda\cap B(x, r)\big), \ r\ge 0.
\end{equation}
Take  discrete sets $\Gamma_\theta\subset T_\theta, \theta\in \Theta$, so that for any $\theta\in \Theta$,
$\{\Phi_{\theta}(\gamma), \gamma\in \Gamma_\theta\}$ forms a minimal phase retrievable frame for $\R^{\# K_\theta}$,
and define
\begin{equation}\label{finite.cor.eq4}
 \Gamma:=\cup_{\theta \in \Theta} \Gamma_\theta,
 \end{equation}
where $\Phi_\theta=(\phi_\lambda)_{\lambda\in K_\theta}$ and
$$K_\theta=\{\lambda\in \Lambda: S_\Phi(\lambda)\cap T_\theta\ne\emptyset\}.$$
Then \eqref{samplingproblem} holds for the above discrete set $\Gamma$. Moreover if
\begin{equation}\label{finite.cor.eq0}
N_{\mathcal T}:=\sup_{\lambda\in \Lambda} \#\{\theta: \ T_\theta\cap  S_\Phi(\lambda)\ne \emptyset\}<\infty,
\end{equation}
then the set $\Gamma$ has finite upper density
\begin{equation}\label{finite.cor.eq3}
D_+(\Gamma)\le  \frac{R_\Lambda(2r_0)(R_\Lambda(2r_0)+1)}{2} N_{\mathcal T} D_+(\Lambda),\end{equation}
 where $r_0$  is given in  \eqref{generator.assumption.eq2}.    \end{theorem}

As an application of Theorem \ref{samplingset2.thm}, 
we have the following phaseless sampling theorem, which is  established in \cite{YC16, CJS17} for signals residing in a principal shift-invariant space generated by a compactly supported function.

\begin{corollary}   \label{samplingset2.cor} Let $ D, \Lambda,  {\mathcal T}, \Phi, V(\Phi)$ and $\Gamma$ be as in Theorem \ref{samplingset2.thm}. Then
any  signal $f\in V(\Phi)$ with ${\mathcal M}_f=\{\pm f\}$ is determined, up to a sign, from its phaseless samples on the discrete set $\Gamma$ with finite  density.
\end{corollary}

 We  remark that the existence of  discrete sets $\Gamma_\theta, \theta \in \Theta$, in Theorem \ref{samplingset2.thm} follows from
 the local complement property on  $T_\theta, \theta\in \Theta$, for the linear space $V(\Phi)$, by applying the argument in \cite[Theorem A.4]{CJS17}.

  \begin{proposition}\label{lcpandprf.pr}
  Let  the domain $D$, the generator
$\Phi: = (\phi_\lambda)_{\lambda\in \Lambda}$,
the family ${\mathcal T}=\{T_\theta, \theta\in \Theta\}$ of open sets, 
and the linear space $V(\Phi)$  be as in Theorem \ref{pr.thm}.
Assume that $\Phi$ has local linear independence on open sets $T_\theta, \theta\in \Theta$. Then  for any $\theta\in \Theta$,
the linear space $V(\Phi)$  generated by $\Phi$
    has local complement property on  $T_\theta$ if and only if there exists  a finite set $\Gamma_\theta\subset T_\theta$ such that
     $\{\Phi_\theta(\gamma), \gamma\in \Gamma_\theta\}$ is a  minimal phase retrievable frame for $\R^{\# K_\theta}$.
\end{proposition}

 We finish this section with the proof of Theorem \ref{samplingset2.thm}.

      \begin{proof}[Proof of Theorem \ref{samplingset2.thm}] 
           First we prove  \eqref{samplingproblem}.
      By \eqref{mfgamma.eq2}, it suffices to prove
     \begin{equation}\label{samplingset2.thm.pf.eq0}
     {\mathcal M}_{f, \Gamma} \subset {\mathcal M}_f.
     \end{equation}
Take  $g=\sum_{\lambda\in \Lambda}d_\lambda\phi_\lambda\in {\mathcal M}_{f, \Gamma}$,  and write $f=\sum_{\lambda\in \Lambda}c_\lambda\phi_\lambda$. Then for any $\theta\in \Theta$,
$$\Big|\sum_{\lambda\in K_\theta} c_\lambda \phi_\lambda(\gamma)\Big|=|f(\gamma)|=|g(\gamma)|=
\Big|\sum_{\lambda\in K_\theta} d_\lambda \phi_\lambda(\gamma)\Big|\ \  {\rm for \ all}\  \gamma\in \Gamma_\theta.$$
This together with the phase retrieval frame property of $\Phi_\theta(\gamma), \gamma\in \Gamma_\theta$, implies that
\begin{equation}
d_\lambda=\delta_\theta c_\lambda, \  \lambda\in K_\theta
\end{equation}
for some $\delta_\theta\in \{-1, 1\}$. Hence for any $\theta\in \Theta$,
\begin{equation}
|g(x)|=|f(x)|,  \  x\in T_\theta.
\end{equation}
This together with Corollary \ref{decomposition.cor} implies that $g\in {\mathcal M}_f$. This proves \eqref{samplingset2.thm.pf.eq0}.

      To prove \eqref{finite.cor.eq3}, we claim that for any $\theta\in \Theta$,
\begin{equation}\label{finite.cor.pf.eq5}
S_{\Phi}(\lambda, \lambda')\ne\emptyset \ {\rm for \ all} \ \lambda, \lambda'\in K_\theta.
\end{equation}
Suppose on the contrary that the above claim does not hold, then there exist
$\lambda_0, \lambda_0^\prime\in K_\theta$ with $S_{\Phi}(\lambda_0, \lambda_0^\prime)=\emptyset$. Thus
$\phi_{\lambda_0}\pm \phi_{\lambda_0^\prime}\in V(\Phi)$ have the same  magnitude measurements  on $T_\theta$, which contradicts to the local complement property of the space $V(\Phi)$ on $T_\theta, \theta\in \Theta$.

Applying  Claim \eqref{finite.cor.pf.eq5} and Assumption \ref{generator.assumption}, we obtain
\begin{equation}\label{finite.cor.pf.eq6}
B(\lambda, r_0)\cap B(\lambda', r_0)\ne \emptyset \ {\rm for \ all} \ \lambda, \lambda'\in K_\theta.
\end{equation}
This  implies that
\begin{equation}\label{finite.cor.pf.eq7}\#K_\theta\le R_\Lambda(2r_0), \ \theta\in \Theta.\end{equation}
Let
$W_{\theta}$ be the linear space of  symmetric matrices spanned by
outer products  $\Phi_{\theta}(x) (\Phi_{\theta}(x))^T, x\in T_\theta$.
Then
\begin{equation}\label{finite.cor.pf.eq7+}
\dim W_{\theta}\le \frac{\# K_{\theta} (\#  K_\theta+1)}{2}.\end{equation}
 Observe that for any $f\in V(\Phi)$, there exists a unique vector $c_\theta=(c_\lambda)_{\lambda\in K_\theta}$ such that
 $$ |f(x)|^2= c_\theta^T \Phi_{\theta}(x) (\Phi_{\theta}(x))^T c_\theta, \  x\in T_\theta.$$
This together the minimality of the phase retrieval frame $\{\Phi_\theta(\gamma), \gamma\in \Gamma_\theta\}$ implies that
 \begin{equation} \label{samplingdensitysufficiency.thm.pf.eq2--}
 \#\Gamma_{\theta}\le \dim W_{\theta}.
 \end{equation}
Combining \eqref{finite.cor.pf.eq7}, \eqref{finite.cor.pf.eq7+} and \eqref{samplingdensitysufficiency.thm.pf.eq2--}, we obtain
\begin{equation}\label{finite.cor.pf.eq9}\# \Gamma_\theta \le \frac{R_\Lambda(2r_0) (R_\Lambda(2r_0)+1)}{2}\ \ {\rm for \ all} \ \theta\in \Theta.\end{equation}

By the minimality of the phase retrieval frame $\{\Phi_\theta(\gamma), \gamma\in \Gamma_\theta\}$, we have $\Phi_\theta(\gamma)\ne 0$ for all $\gamma\in \Gamma_\theta$, which implies that
\begin{equation}\label{finite.cor.pf.eq8}
\Gamma_\theta\subset \big(\cup_{\lambda\in K_\theta} S_\Phi(\lambda)\big)\cap T_\theta
\end{equation}
Then for any $x\in D$ and $r\ge 0$,
we obtain from  \eqref {finite.cor.eq0}, \eqref{finite.cor.pf.eq9}, \eqref{finite.cor.pf.eq8}   and Assumption \ref{generator.assumption} that
\begin{eqnarray} \label {finite.cor.pf.eq10}
\#(\Gamma\cap B(x, r))   &\hskip-0.08in \le & \hskip-0.08in  \Big(\max_{\theta\in \Theta} \# \Gamma_\theta\Big) \nonumber\\
& & \times
\#\{\theta\in \Theta:  \big(\cup_{\lambda\in K_\theta} S_\Phi(\lambda)\big)\cap T_\theta\cap B(x, r)\ne \emptyset\}\nonumber\\
& \hskip-0.08in \le &  \hskip-0.08in \frac{ R_\Lambda(2r_0) (R_\Lambda(2r_0)+1))}{2}  \Big(\max_{\lambda\in \Lambda} \#\{\theta\in \Theta:  S_\Phi(\lambda)\cap T_\theta\ne \emptyset\}\Big)\nonumber\\
& & \times
\#\{\lambda\in \Lambda: S_\Phi(\lambda)\cap B(x, r)\ne \emptyset\}\nonumber\\
& \hskip-0.08in \le & \hskip-0.08in  \frac{ R_\Lambda(2r_0) (R_\Lambda(2r_0)+1))}{2}   N_{\mathcal T}
\#(\Lambda\cap B(x, r+r_0)).
\end{eqnarray}
This together with  \eqref {density.thm.eq1}  and definition of the density \eqref{generator.assumption.eq1} of a discrete set proves
\eqref{finite.cor.eq3}.
      \end{proof}

\section{Stable Reconstruction  from Phaseless Samples}
\label{stablereconstruction.section}

Let
 ${\mathcal T}=\{T_\theta:\ \theta\in \Theta\}$ satisfy \eqref{Ttheta.def}
 and
 $\Gamma=\cup_{\theta\in \Theta} \Gamma_\theta$ with $\Gamma_\theta\subset T_\theta, \theta\in \Theta$
 be as in Theorem \ref{samplingset2.thm}.
In this section,  we propose the following three-step algorithm, MAPS for abbreviation,
  to construct an approximation
  \begin{equation}\label{geta.def}
  g_\eta=\sum_{\lambda\in \Lambda} d_{\eta; \lambda} \phi_\lambda\end{equation}
  to the original signal $f\in V(\Phi)$ in   magnitude measurements  from its noisy phaseless samples
 \begin{equation}\label{zepsilon.def}
 z_\eta(\gamma)= |f(\gamma)|+\eta(\gamma),\  \gamma\in \Gamma,
 \end{equation}
  taken on a discrete set $\Gamma$ and
 corrupted by
  a bounded noise  $\eta=(\eta(\gamma))_{\gamma\in \Gamma}$.
  \smallskip

\begin{tcolorbox}[colback=white,colframe=black]
 \begin{itemize}
 \item [{0.}] Select a phase adjustment threshold value $M_0\ge 0$ and set
 $K_\theta=\{\lambda\in \Lambda: \ S_\Phi(\lambda)\cap T_\theta\ne\emptyset\}$. 
\item[{1.}]
For $\theta \in \Theta$,
 let \begin{equation}\label{lomas.def1}
{c}_{\eta, \theta}=(c_{\eta, \theta; \lambda})_{\lambda\in \Lambda}
\end{equation}
take zero components except that
 $(c_{\eta, \theta; \lambda})_{ \lambda \in K_{\theta}}$ is a   solution of the  local  {\bf m}inimization problem
\begin{eqnarray} \label{lomas.def2}
& & \min_{(d_\lambda)_{ \lambda\in K_{\theta}}}  \sum_{{ \gamma} \in \Gamma_\theta} \Big|\Big|\sum_{\lambda \in K_{\theta}} d_\lambda \phi_\lambda({\gamma})\Big|-z_{\eta}({ \gamma} )\Big|^2
\nonumber\\
 & = & \min_{\delta_{\gamma} \in \{-1, 1\}, \gamma\in \Gamma_\theta}
 \min_{(d_\lambda)_{ \lambda\in K_{\theta}}}  \sum_{{ \gamma} \in \Gamma_\theta} \Big|\sum_{\lambda \in K_{\theta}} d_\lambda \phi_\lambda({\gamma})-\delta_\gamma z_{\eta}({ \gamma} )\Big|^2.
%
\end{eqnarray}

\item [{2.}] {\bf A}djust  {\bf p}hases of vectors ${c}_{\eta, \theta}, \theta\in \Theta$, so that
the resulting vectors $ \delta_{\eta, \theta}{ c}_{\eta, \theta}$
with $\delta_{\eta, \theta}\in \{-1, 1\}$ have their inner product at least $-M_0$,
\begin{equation} \label{lomas.def3}
 \langle   \delta_{\eta, \theta} { c}_{\eta, \theta},  \delta_{\eta, \theta'} { c}_{\eta,\theta'}\rangle=
 \delta_{\eta, \theta} \delta_{\eta, \theta'}\sum_{\lambda\in K_\theta\cap K_{\theta'}}
 { c}_{\eta, \theta; \lambda} { c}_{\eta, \theta'; \lambda}
 \ge -M_0 
\end{equation}
 for   all $ \theta, \theta' \in \Theta$.
 \item[{3.}] {\bf S}ew  vectors $\delta_{\eta, \theta}{c}_{\eta, \theta}, \theta\in \Theta$, together to obtain
\begin{equation}\label{lomas.def4}
d_{\eta; \lambda}=\frac{\sum_{\theta\in \Theta } \delta_{\eta, \theta}{ c}_{\eta, \theta; \lambda}\chi_{K_\theta}(\lambda) }{\sum_{\theta\in \Theta }\chi_{K_\theta}(\lambda)}, \  \lambda\in \Lambda,
\end{equation}
where $\chi_E$ is the indicator function on a set $E$.
 \end{itemize}
\end{tcolorbox}

The prior versions of the above MAPS algorithm are used in \cite{YC16, CJS17} to reconstruct signals in a principal shift-invariant space from their noisy phaseless samples. As shown in the following remark that complexity of the proposed  MAPS algorithm depends almost linearly on the size of
the original signal.

\begin{remark} {\rm Take a signal
$f=\sum_{\lambda\in \Lambda_0} c_\lambda \phi_\lambda\in V(\Phi)$ with component vector $(c_\lambda)_{\lambda\in \Lambda_0}$ supported in  $\Lambda_0\subset \Lambda$, and
define
$\Theta_0=\{\theta\in \Theta: K_\theta\cap \Lambda_0\ne \emptyset\}$. By
\eqref{lomas.def4}, in the first step of the proposed MAPS algorithm,  it suffices to solve  local minimization problems \eqref{lomas.def2} with $\theta\in \Theta_0$.
Observe that
\begin{equation}\label{complexity.rem.eq1}
\# \Theta_0=\#\big(\cup_{\lambda\in \Lambda_0} \{\theta\in \Theta, \lambda\in K_\theta\}\big)
\le  N_{\mathcal T}   N  
\end{equation}
by \eqref{finite.cor.eq0}, where  $N=\#\Lambda_0$  is the size of supporting component vector
 of the original signal $f$.  This together with \eqref{finite.cor.pf.eq7} and \eqref{finite.cor.pf.eq9}
implies that
 the number of additions and multiplications required in the first step  is $O(N)$.
  By \eqref{lomas.def4}, in the second step it suffices to verify the phase adjustment condition \eqref{lomas.def3}
  for all $\theta, \theta'\in \Theta_0$ with $K_\theta\cap K_{\theta'}\ne \emptyset$.
     For any $\theta\in \Theta$, we obtain from
   \eqref{finite.cor.eq0}  and \eqref{finite.cor.pf.eq7} that
    \begin{eqnarray}\label{complexity.rem.eq2}
    \#\{\theta'\in \Theta: \ K_\theta\cap K_{\theta'}\ne \emptyset\} & \hskip-0.08in \le & \hskip-0.08in \# \big( \cup_{\lambda\in K_\theta} \{\theta'\in \Theta: \ \lambda\in K_{\theta'}\}\big)\nonumber\\
    & \hskip-0.08in \le & \hskip-0.08in N_{\mathcal T} \# K_\theta\le   N_{\mathcal T} R_\Lambda(2r_0).
    \end{eqnarray}
 Therefore the number of additions and multiplications required in the second step  is $O(N)$
 by  \eqref{finite.cor.pf.eq7}, \eqref{complexity.rem.eq1} and \eqref{complexity.rem.eq2}.
 By \eqref{finite.cor.eq0}, the number of additions and multiplications required in the third step of the proposed MAPS algorithm is $O(N)$.
 Combining the above arguments, we conclude that the number of additions and multiplications required in
 the proposed MAPS algorithm to reconstruct an approximation  $g_\eta$ of the original signal $f$ is about $O(N)$.
}\end{remark}

 For a bounded signal $f$ on the domain $D$, we denote its $L^\infty$ norm  by $\|f\|_\infty:=\sup_{x\in D} |f(x)|$,
 and for  a phase retrievable frame ${\mathcal F}=\{f_m\in \R^n, 1\le m\le M\}$, we use 
\begin{eqnarray} \label{phil2.def}
\big\|{\mathcal F}\|_{\rm P} & = & \inf_{T \subset \{1, \cdots, M\}}
 \max\Bigg(
\inf_{ \|v\|_2=1} \Big(\sum_{m\in T} |\langle v, f_m\rangle |^2\Big)^{1/2},\nonumber\\
& & \qquad \qquad \qquad\qquad
\inf_{ \|v\|_2=1} \Big(\sum_{m\not\in T} |\langle v, f_m\rangle |^2\Big)^{1/2} \Bigg) 
\end{eqnarray}
to describe the stability to reconstruct a vector $v$ from its phaseless frame measurements $|\langle v, f_m\rangle|, 1\le m\le M$.
In the next theorem, we show that
  the signal  $g_\eta$ reconstructed from the proposed  MAPS algorithm
 with the
 phase adjustment threshold value $M_0$ properly chosen
provides an approximation to the original signal in magnitude measurements.

\begin{theorem}\label{stability.thm} Let  the domain $D$, the generator
$\Phi: = (\phi_\lambda)_{\lambda\in \Lambda}$,
the family ${\mathcal T}=\{T_\theta, \theta\in \Theta\}$ of open sets, 
and the linear space $V(\Phi)$  be as in Theorem \ref{pr.thm}.
Assume that
the generator $\Phi$ is uniformly bounded in the sense that \begin{equation}
\|\Phi\|_\infty:=\sup_{\lambda\in \Lambda} \|\phi_\lambda\|_\infty<\infty,
\end{equation}
and the sampling set $\Gamma=\cup_{\theta \in \Theta} \Gamma_\theta$
are so chosen that
$\Gamma_\theta\subset T_\theta$ for all  $\theta\in \Theta$,
$\Phi_{\theta, \Gamma_\theta}=\{\Phi_\theta(\gamma), \gamma\in \Gamma_\theta\}, \theta\in \Theta$, are  phase retrievable frames, and
  \begin{equation} \max_{\theta\in \Theta} \# \Gamma_\theta  (\|\Phi_{\theta, \Gamma_\theta}\|_{P})^{-2}<\infty.
 \end{equation}
Given a signal $f\in V(\Phi)$
and a bounded noise  $\eta=(\eta(\gamma))_{\gamma\in \Gamma}$, let
$g_\eta$ be the reconstructed signal from noisy phaseless samples $z_\eta(\gamma), \gamma\in \Gamma$ in \eqref{zepsilon.def}
via the MAPS algorithm
 \eqref{geta.def}--\eqref{lomas.def4},
where \begin{equation}  \label{stability.thm.eq1}
M_0= 24 \Big(\max_{\theta\in \Theta} \# \Gamma_\theta  \big(\|\Phi_{\theta, \Gamma_\theta}\|_{P}\big)^{-2} \Big) \|\eta\|_\infty^2 
\end{equation}
and
\begin{equation}\label{etainfinitenorm}
\|\eta\|_\infty:=\sup_{\gamma\in \Gamma} |\eta(\gamma)|<\infty.\end{equation}
 Then
there exist   $f_\eta, h_\eta\in V(\Phi)$
with the same magnitude measurements on the whole domain,
 \begin{equation}  \label{stability.thm.eq3}
{\mathcal M}_{h_\eta}={\mathcal M}_{f_\eta},
\end{equation}
which are
approximations
to the original signal $f$ and
the reconstruction $g_\eta$  respectively,
\begin{equation} \label{stability.thm.eq4}
\|f_\eta-f\|_\infty \le 4\sqrt{6} \Big(\max_{\theta\in \Theta} \sqrt{\# \Gamma_\theta}  \big(\|\Phi_{\theta, \Gamma_\theta}\|_{P}\big)^{-1} \Big)
  R_\Lambda(r_0)  \|\Phi\|_\infty
 \|\eta\|_\infty
\end{equation}
and
\begin{equation} \label{stability.thm.eq5}
\|g_\eta-h_\eta\|_\infty \le 6\sqrt{6} \Big(\max_{\theta\in \Theta} \sqrt{\# \Gamma_\theta}  \big(\|\Phi_{\theta, \Gamma_\theta}\|_{P}\big)^{-1} \Big)
  R_\Lambda(r_0)  \|\Phi\|_\infty
 \|\eta\|_\infty.
\end{equation}
 \end{theorem}

 In the  noiseless environment (i.e. $\eta=0$), it follows from
Theorem \ref{stability.thm} that
  the signal reconstructed from the   MAPS algorithm
 with
 phase adjustment threshold value $M_0=0$ has the same magnitude measurements on the whole domain as the original signal, cf. Theorem
  \ref{generalphaseretrieval.thm}.

  \smallskip

By Theorem
  \ref{stability.thm},  we obtain
\begin{eqnarray}\label{stability.thm.eq2*}\big\||g_\eta|-|f|\big\|_\infty
 & \hskip-0.08in \le & \hskip-0.08in
\|g_\eta-h_\eta\|_\infty+ \|f-f_\eta\|_\infty\nonumber\\
& \hskip-0.08in \le & \hskip-0.08in 10
\sqrt{6} \Big(\max_{\theta\in \Theta} \sqrt{\# \Gamma_\theta}  \big(\|\Phi_{\theta, \Gamma_\theta}\|_{P}\big)^{-1} \Big)
  R_\Lambda(r_0)  \|\Phi\|_\infty
 \|\eta\|_\infty.
\end{eqnarray}
Take $\lambda_0\in \Lambda$  so that $\|\phi_{\lambda_0}\|_\infty \ge \|\Phi\|_\infty/2$.
Then for any signal $f\in V(\Phi)$ and $\epsilon\ge 0$,
we have
\begin{equation}\label{stability.thm.eq3*}
 \big||f(\gamma)\pm \epsilon  \phi_{\lambda_0}(\gamma)|-|f(\gamma)|\big|\le  \|\Phi\|_\infty \epsilon, \ \gamma\in \Gamma\end{equation}
and
\begin{eqnarray} \label{stability.thm.eq4*}  & & \max\Big(\big\||f+\epsilon  \phi_{\lambda_0}|-|f|\big\|_\infty, \big\||f-\epsilon  \phi_{\lambda_0}|-|f|\big\|_\infty\Big)\nonumber\\
& \hskip-0.08in =&\hskip-0.08in
\Big\|\max\Big( \big||f+\epsilon  \phi_{\lambda_0}|-|f|\big|,  \big| |f-\epsilon  \phi_{\lambda_0}|-|f|\big|\Big)\Big\|_\infty\nonumber\\
& \hskip-0.08in \ge &\hskip-0.08in \|\epsilon \phi_{\lambda_0}\|_\infty \ge \frac{1}{2} \|\Phi\|_\infty \epsilon.
\end{eqnarray}
 By \eqref{stability.thm.eq2*}, \eqref{stability.thm.eq3*} and \eqref{stability.thm.eq4*}, we conclude that the reconstructed  signal
$ g_{\eta}$ from the proposed MAPS algorithm
  is a suboptimal approximation to the original signal $f$ in magnitude measurements.

\smallskip

Take $g\in V(\Phi)$.
For the noise  $\eta=(\eta(\gamma))_{\gamma\in \Gamma}$ in \eqref{zepsilon.def} given by
 $\eta(\gamma)= |g(\gamma)|-|f(\gamma)|, \gamma\in \Gamma$, one may verify that the
  signal $g_\eta$ reconstructed from the MAPS algorithm
  could have  the same magnitude measurements as the  signal $g$ has, i.e., $g_\eta\in {\mathcal M}_g$.
This together with  \eqref{stability.thm.eq2*} leads to the bi-Lipschitz property for the phaseless sampling operator on $V(\Phi)$.

\begin{corollary}\label{phaselesssamplingstability.cor}
Let  the domain $D$, the generator
$\Phi$,
the family ${\mathcal T}$ of open sets, 
the phaseless sampling set $\Gamma$,
and the linear space $V(\Phi)$  be as in Theorem \ref{stability.thm}.
Then the phaseless sampling operator
$$S: V(\Phi)\ni f\longmapsto (|f(\gamma)|)_{\gamma\in \Gamma}$$
is bi-Lispchitz in magnitude measurements,
i.e., there exist positive constants $C_1$ and $C_2$ such that
\begin{equation}\label{bilichitz}
C_1\big\||g|-|f|\big\|_\infty\le  \|Sf- Sg\|_\infty \le  C_2  \big\||g|-|f|\big\|_\infty
\end{equation}
for all signals $f, g\in V(\Phi)$.
\end{corollary}


We finish this section with the proof of Theorem \ref{stability.thm}.

\begin{proof} [Proof of Theorem \ref{stability.thm}]
Take $\theta\in \Theta$
and  define
\begin{equation}\label{stability.thm.pf.eq0} g_{\eta, \theta}=\sum_{\lambda\in \Lambda} c_{\eta, \theta; \lambda}\phi_\lambda,\end{equation}
where $c_{\eta, \theta; \lambda}, \lambda\in \Lambda$, are given in \eqref{lomas.def1}.
Then  there exists  a subset $\Gamma_\theta^\prime\subset \Gamma_\theta$ such that
\begin{eqnarray}\label{stability.thm.pf.eq1}  
 & &\Big(\sum_{\gamma\in \Gamma_\theta^\prime} \big|g_{\eta, \theta}(\gamma)- f(\gamma)\big|^2\Big)^{\frac 12}
 +
\Big(\sum_{\gamma\in \Gamma_\theta\backslash \Gamma_\theta^\prime} \big|g_{\eta, \theta}(\gamma)+ f(\gamma)\big|^2\Big)^{\frac 12} \nonumber\\
 &\hskip-0.08in = & \hskip-0.08in \Big(\sum_{\gamma\in \Gamma_\theta^\prime} \big| |g_{\eta, \theta}(\gamma)|- |f(\gamma)|\big|^2\Big)^{\frac 12}
 +
\Big(\sum_{\gamma\in \Gamma_\theta\backslash \Gamma_\theta^\prime} \big||g_{\eta, \theta}(\gamma)|-|f(\gamma)|\big|^2\Big)^{\frac 12} \nonumber\\
& \hskip-0.08in  \le  & \hskip-0.08in \sqrt{2} \Big(\sum_{\gamma\in \Gamma_\theta} \big||g_{\eta, \theta}(\gamma)|-|f(\gamma)|\big|^2\Big)^{\frac 12}
 \nonumber\\
 & \hskip-0.08in \le &\hskip-0.08in
 \sqrt{2}\Big(\sum_{\gamma\in \Gamma_\theta}\big| |g_{\eta, \theta}(\gamma)|-
 z_{\eta} (\gamma)\big|^2\Big)^{\frac 12} + \sqrt{2} \Big(\sum_{\gamma\in \Gamma_\theta}\big|
 |f(\gamma)|-z_{\eta} (\gamma)\big|^2\Big)^{\frac 12}\nonumber\\
\hskip0.21in & \hskip-0.08in \le & \hskip-0.08in 2\sqrt{2}
 \Big(\sum_{\gamma\in \Gamma_\theta}\big|
 |f(\gamma)|-z_{\eta} (\gamma)\big|^2\Big)^{\frac 12}\le  2\sqrt{2} \sqrt{\# \Gamma_\theta}  \|\eta\|_{\infty},
\end{eqnarray}
where the third inequality follows from \eqref{lomas.def2}
and the last inequality holds by   \eqref{zepsilon.def}.
By \eqref{lomas.def1} and the definitions of the sets $K_\theta$ and $\Gamma_\theta, \theta\in \Theta$, we have
\begin{equation}\label{stability.thm.pf.eq1+}
g_{\eta, \theta}(\gamma)\pm f(\gamma)=\sum_{\lambda\in K_\theta} (c_{\eta, \theta; \lambda}\pm  c_\lambda) \phi_\lambda (\gamma), \ \gamma\in \Gamma_\theta.\end{equation}
By \eqref{phil2.def}, \eqref{stability.thm.pf.eq1},
\eqref{stability.thm.pf.eq1+}
and the phase   retrievable frame assumption for $\Phi_{\theta, \Gamma_\theta}$, we obtain that
\begin{equation} \label{stability.thm.pf.eq2}
 \Big(\sum_{ \lambda \in K_{\theta} } |c_{\eta, \theta; \lambda}- \tilde \delta_{\eta, \theta} c_\lambda|^2\Big)^{1/2} 
  \le   2\sqrt{2}  \sqrt{\# \Gamma_\theta} \big(\|\Phi_{\theta, \Gamma_\theta}\|_{\rm P}\big)^{-1}\|\eta\|_\infty
\end{equation}
for some $\tilde \delta_{\eta, \theta}\in \{-1, 1\}$.

Let $\tilde \delta_{\eta, \theta}, \theta\in \Theta$, be as in \eqref{stability.thm.pf.eq2}.
Then for any $\theta, \theta'\in \Theta$,  we have
\begin{eqnarray}\label{stability.thm.pf.eq2+}
  \langle \tilde\delta_{\eta, \theta}{ c}_{\eta, \theta},
\tilde\delta_{\eta, \theta'}{c}_{\eta, \theta'}\rangle
  & \hskip-0.08in = &  \hskip-0.08in \sum_{\lambda \in K_{\theta}\cap K_{\theta'}}
 \tilde \delta_{\eta, \theta}
\tilde\delta_{\eta, \theta'}
{ c}_{\eta, \theta;\lambda}{c}_{\eta, \theta';\lambda}
 \nonumber\\
    & \hskip-0.08in \ge  & \hskip-0.08in     \sum_{\lambda\in K_{\theta}\cap K_{\theta'}}|c_\lambda|^2
-
  \sum_{\lambda\in K_{\theta}\cap K_{\theta'}}
|c_\lambda|  |  \tilde\delta_{\eta, \theta}c_{ \eta, \theta; \lambda}-c_\lambda| \nonumber\\
   & & \hskip-0.08in - \sum_{\lambda\in K_{\theta}\cap K_{\theta'}}
 |\tilde \delta_{\eta, \theta'} c_{\eta,  \theta'; \lambda}- c_\lambda|
|c_\lambda|\nonumber\\
& & \hskip-0.08in - \sum_{\lambda\in K_{\theta}\cap K_{\theta'}}
 |\tilde\delta_{\eta, \theta} c_{ \eta, \theta; \lambda}- c_\lambda|
  |\tilde \delta_{\eta, \theta'}c_{\eta, \theta'; \lambda}-c_\lambda|\nonumber\\
   & \hskip-0.08in\ge  & \hskip-0.08in  \frac{1}{2} \sum_{\lambda\in K_{\theta}\cap K_{\theta'}}|c_\lambda|^2 \nonumber\\
   \hskip-0.08in & \hskip-0.08in & \hskip-0.08in  -
  \frac{3}{2} \sum_{\lambda\in K_{\theta}\cap K_{\theta'}}
  \Big(| \tilde\delta_{\eta, \theta}c_{\eta, \theta; \lambda}-c_\lambda|^2 +
 |\tilde\delta_{\eta, \theta'} c_{\eta,  \theta'; \lambda}- c_\lambda|^2\Big).
 \end{eqnarray}
 This together with \eqref{stability.thm.eq1} and \eqref{stability.thm.pf.eq2}   implies
\begin{equation}\label{stability.thm.pf.eq3}
\langle \tilde\delta_{\eta, \theta}{ c}_{\eta, \theta},
\tilde\delta_{\eta, \theta'}{c}_{\eta, \theta'}\rangle \ge   -24  \# \Gamma_\theta  \big(\|\Phi_{\theta, \Gamma_\theta}\|_{\rm P}\big)^{-2}
 \|\eta\|_\infty^2 \ge -M_0
\end{equation}
for all $\theta, \theta'\in \Theta$.
This proves that phases of  ${ c}_{\eta, \theta}, \theta \in \Theta$, in \eqref{lomas.def1}
can be adjusted so that \eqref{lomas.def3} holds.

Let $\delta_{\eta, \theta}\in \{-1, 1\}, \theta\in \Theta$, be signs  in \eqref{lomas.def3} used for the phase adjustment
of vectors ${ c}_{\eta, \theta}, \theta \in \Theta$, in \eqref{lomas.def1}.
We remark that the above signs are not necessarily the ones in \eqref{stability.thm.pf.eq2}, however  as shown in  \eqref{stability.thm.pf.eq7}  below they are related.
Define
\begin{equation}\label{stability.thm.pf.eq3+0}
f_\eta=\sum_{|c_\lambda|> 2\sqrt{M_0}} c_\lambda \phi_\lambda.
\end{equation}
Then  for $x\in D$, we obtain from \eqref{generator.assumption.eq2} and \eqref{rs.def}  that
\begin{equation*}
 |f(x)-f_\eta(x)|
 \le    2 \sqrt{M_0}
\sum_{\lambda\not\in V_{f_\eta}}  |\phi_\lambda(x)| \le    2 \sqrt{M_0}   R_\Lambda(r_0)  \|\Phi\|_\infty,
\end{equation*}
which proves \eqref{stability.thm.eq4}.

By \eqref{stability.thm.eq1}, \eqref{stability.thm.pf.eq2} and
\eqref{stability.thm.pf.eq2+}, we obtain that
\begin{equation}\label{stability.thm.pf.eq3+}
\langle \tilde\delta_{\eta, \theta}{ c}_{\eta, \theta},
\tilde\delta_{\eta, \theta'}{c}_{\eta, \theta'}\rangle >   M_0
\end{equation}
for all $\theta, \theta'\in \Theta$ with $K_\theta\cap K_{\theta'}\cap V_{f_\eta}\ne \emptyset$.
This together with \eqref{lomas.def3} implies that
\begin{equation*}
\delta_{\eta, \theta}\tilde \delta_{\eta, \theta}=
\delta_{\eta, \theta'} \tilde\delta_{\eta, \theta'}
\end{equation*}
hold for all pairs $(\theta, \theta')$ satisfying
$K_\theta\cap K_{\theta'}\cap V_{f_\eta}\ne \emptyset$.
Hence for $\lambda\in V_{f_\eta}$ there exists
$\delta_\lambda\in \{-1, 1\}$ such that
\begin{equation} \label{stability.thm.pf.eq5}
\delta_{\eta, \theta}\tilde \delta_{\eta, \theta}=\delta_\lambda
\end{equation}
for all $\theta\in \Theta$ satisfying $\lambda \in K_\theta$.
Decompose the graph ${\mathcal G}_{f_\eta}$ into the union of connected components $(V_{\eta, i}, E_{\eta, i}), i\in I_\eta$,
and  the signal $f_\eta$ as in \eqref{decomposition.def1}, \eqref{decomposition.def2} and \eqref{decomposition.def3},
 \begin{equation}\label{feta.demposition}
 f_\eta=\sum_{i\in I_\eta} \sum_{\lambda\in V_{\eta, i}} c_\lambda \phi_\lambda.
 \end{equation}
Observe that for any edge $(\lambda, \lambda')$ of $V_{f_\eta}$, there exists $\theta_0\in \Theta$ such that $\lambda, \lambda'\in K_{\theta_0}$ by
\eqref{Ttheta.def}. Hence
\begin{equation} \label{stability.thm.pf.eq6}
\delta_\lambda= \delta_{\eta, \theta_0}\tilde \delta_{\eta, \theta_0}=\delta_{\lambda'}.
\end{equation}
Combining \eqref{stability.thm.pf.eq5} and \eqref{stability.thm.pf.eq6}, there exists $\delta_i, i\in I_\eta$, such that
\begin{equation} \label{stability.thm.pf.eq7}
\delta_{\eta, \theta}\tilde \delta_{\eta, \theta}=\delta_i
\end{equation}
for all $\theta\in \Theta$ satisfying $K_\theta\cap V_{\eta,i}\ne \emptyset$.
Set
$$h_{\eta}=\sum_{i\in I_\eta}\delta_i \sum_{\lambda\in V_{\eta, i}} c_\lambda \phi_\lambda.$$
Then $f_\eta$ and $h_\eta$ have the same magnitude measurements on the whole domain by
 \eqref{decomposition.def1}, which proves \eqref{stability.thm.eq3}.

For all $\lambda\not\in V_{f_\eta}$, we obtain from \eqref{stability.thm.pf.eq2} that
\begin{equation}    \label{stability.thm.pf.eq8--}
|d_{\eta, \lambda}| \le
\frac{\sum_{ K_\theta \ni \lambda } (|\delta_{\eta, \theta} c_{\eta, \theta; \lambda}- \delta_{\eta, \theta} \tilde \delta_{\eta,\theta} c_\lambda|+|c_\lambda|)}
{\sum_{  K_\theta \ni \lambda}  1}
 \le  3\sqrt{M_0}.
\end{equation}
For any $\lambda\in V_{\eta, i}, i\in I_\eta$,  we get
\begin{eqnarray}
  \label{stability.thm.pf.eq8}
|d_{\eta,\lambda}-\delta_i c_{\lambda}|
 & \le &
\frac{\sum_{  K_\theta \ni \lambda }|  \delta_{\eta, \theta} { c}_{\eta, \theta; \lambda}-\delta_i { c}_\lambda| }{\sum_{ K_\theta \ni  \lambda }
1 }\nonumber\\
& = &
\frac{\sum_{ K_\theta \ni \lambda }|  { c}_{\eta, \theta; \lambda}-\tilde \delta_{\eta, \theta}{c}_\lambda| }{\sum_{K_\theta \ni   \lambda } 1}
 \le \sqrt{M_0}.
\end{eqnarray}
Combining  \eqref{stability.thm.pf.eq8--} and \eqref{stability.thm.pf.eq8}, we obtain
\begin{eqnarray} \label{stability.thm.pf.eq9}
|g_\eta(x)-h_\eta(x)| & \hskip-0.08in \le &  \hskip-0.08in \sum_{\lambda\not\in V_{f_\eta}} |d_{\eta, \lambda}||\phi_\lambda(x)|
+\sum_{i\in I_\eta} \sum_{\lambda\in V_{\eta, i}} |d_{\eta, \lambda}-\delta_i c_\lambda| |\phi_\lambda(x)|\nonumber\\
& \hskip-0.08in \le & \hskip-0.08in 3\sqrt{M_0} \sum_{\lambda\in \Lambda} |\phi_\lambda(x)|\le 3 \sqrt{M_0} \|\Phi\|_\infty \sum_{\lambda\in \Lambda} \chi_{B(\lambda, r_0)}(x)
\nonumber\\
& \hskip-0.08in \le & \hskip-0.08in  3 \sqrt{M_0}   R_\Lambda(r_0)  \|\Phi\|_\infty \ \ {\rm for \ all}\ x\in D,
\end{eqnarray}
which proves \eqref{stability.thm.eq5}.
This  completes the proof.
\end{proof}

\section{Numerical Simulations}
\label{simulation.section}

In this section, we present some numerical simulations to demonstrate the performance of the MAPS algorithm proposed in the last section,
where  signals are
one-dimensional non-uniform cubic splines and two-dimensional piecewise affine functions on a triangulation. 

Denote the positive part of a real number $x$ by $x_+=\max(x, 0)$. In the first simulation, we consider phaseless sampling and reconstruction of  cubic spline signals $f$ on the interval $[a, b]$  with non-uniform knots $a=t_0<t_1<\ldots<t_N=b$, see the left image of Figure \ref{spline_separable.fig} 
where $a=0, b=100$ and $N=100$. Those  signals have the following
parametric representation
\begin{equation}\label{spline.def0}
f(x)= \sum_{n=0}^{N-4} c_n B_{n}(x),\  x\in [a, b],
\end{equation}
 where
 \begin{equation*}
\label{B_spline.eq}
 B_{n}(x)=(t_{n+4}-t_n)\sum_{l=0}^4 \frac{(x-t_{n+l})_+^3}{\prod_{0\le j\le 4, j\ne l} (t_{n+l}-t_{n+j})}, \ \ 0\le n\le N-4
 \end{equation*}
 are cubic B-splines with knots $t_{n+l}, 0\le l\le 4$   \cite{unser99, wahba90}.
 In our simulations, we assume that
$$c_n\in [-1, 1], \  0\le n\le N-4,$$
 are randomly selected,   and
$$t_n= a+  (n+\epsilon_n)\frac{b-a}{N},\ 1\le n\le N-1$$
 for some  $\epsilon_n, 1\le n\le N-1$, being randomly selected  in $[-0.2, 0.2]$.
Then cubic spline signals in the first simulation  
 have $(b-a)/N$ as their rate of innovation.

 Consider the scenario that
phaseless samples of the signal  $f$ in \eqref{spline.def0} on a discrete 
 set
$\Gamma$ are corrupted
by a bounded random noise,
\begin{equation}\label{noisydata2.def}
{ z}_{\eta}({\gamma})= |f({\gamma})|+\eta({ \gamma}), \ {\gamma}\in \Gamma,
\end{equation}
 where $ \eta(\gamma),  \gamma\in \Gamma$,  are randomly selected in the interval $[-\eta, \eta]$ for  some $\eta \ge 0$,
 \begin{equation}\label{spline.samplingset}
 \Gamma:=\cup_{n=0}^{N-1} \Gamma_n:= \bigcup_{n=0}^{N-1} \Big\{t_n+ k \frac{t_{n+1}-t_n}{K+1}\in (t_n, t_{n+1}),\  1\le k\le K\Big\},\end{equation}
     and
  $K\ge 7$ is a  positive integer.

Denote  by $g_\eta$ the reconstructed signal from the above noisy phaseless samples via the proposed MAPS algorithm. Presented on the  top left and right of Figure \ref{spline_quadratic.pic} 
are the reconstructed signal $g_\eta$ via the proposed MAPS algorithm
and the difference  $|g_\eta|-|f_o|$
  between magnitudes of the reconstructed signal  $g_\eta$ and the original signal  $f_o$  plotted on the left  of Figure
    \ref{spline_separable.fig}  respectively, 
     where $\eta=0.01, K=9$ and
the  maximal  error
   $\||g_\eta|-|f_o|\|_\infty$  in
    magnitude measurements  is $0.2104$.
This demonstrates the approximation property in Theorem \ref{stability.thm}. 
Unlike four ``islands"
  decomposition
 \eqref{decomposition.def1}, \eqref {decomposition.def2} and \eqref{decomposition.def3} for the original signal $f_o$, signals $f_\eta$
 and $h_\eta$   used to approximate the original signal $f_o$ and the reconstructed signal $g_\eta$
 in  Theorem \ref{stability.thm}
 have
 five ``islands"
  decomposition
 \eqref{decomposition.def1}, \eqref {decomposition.def2} and \eqref{decomposition.def3}, see the bottom left of Figure
    \ref{spline_quadratic.pic}.
 \begin{figure}[t] 
\begin{center}
\includegraphics[width=60mm, height=38mm]{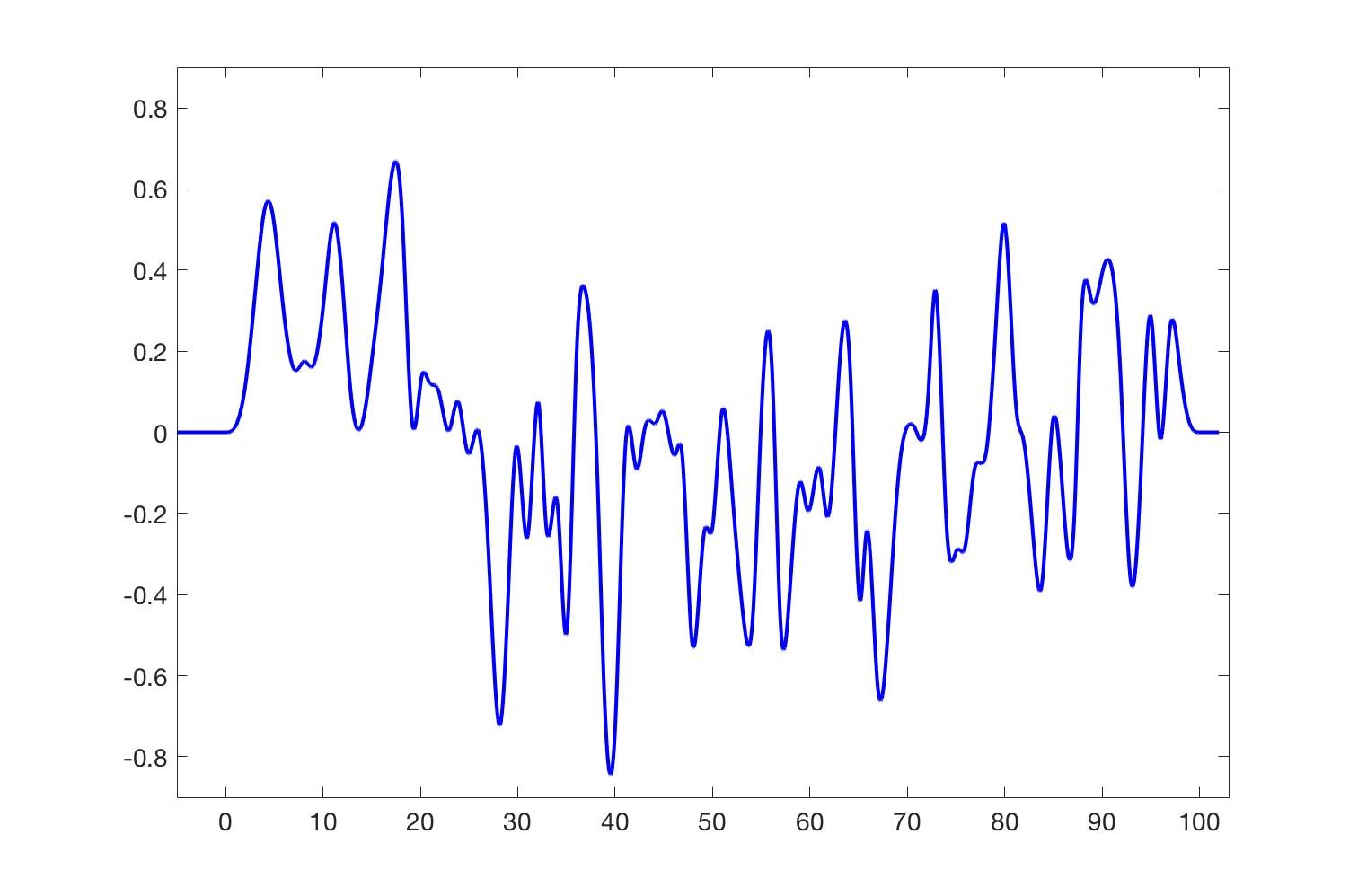}
\includegraphics[width=60mm, height=38mm]{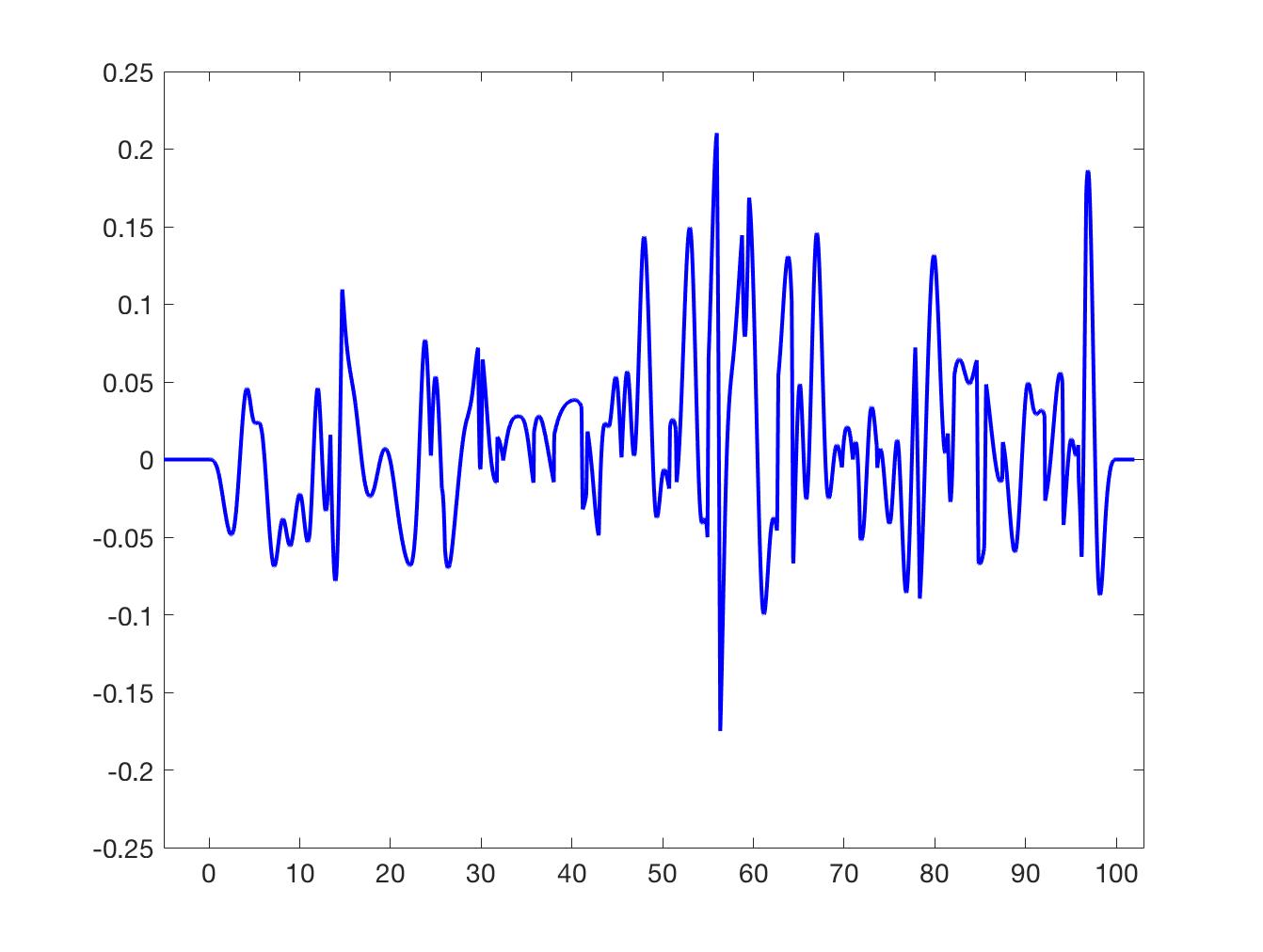}\\
\includegraphics[width=60mm, height=38mm]{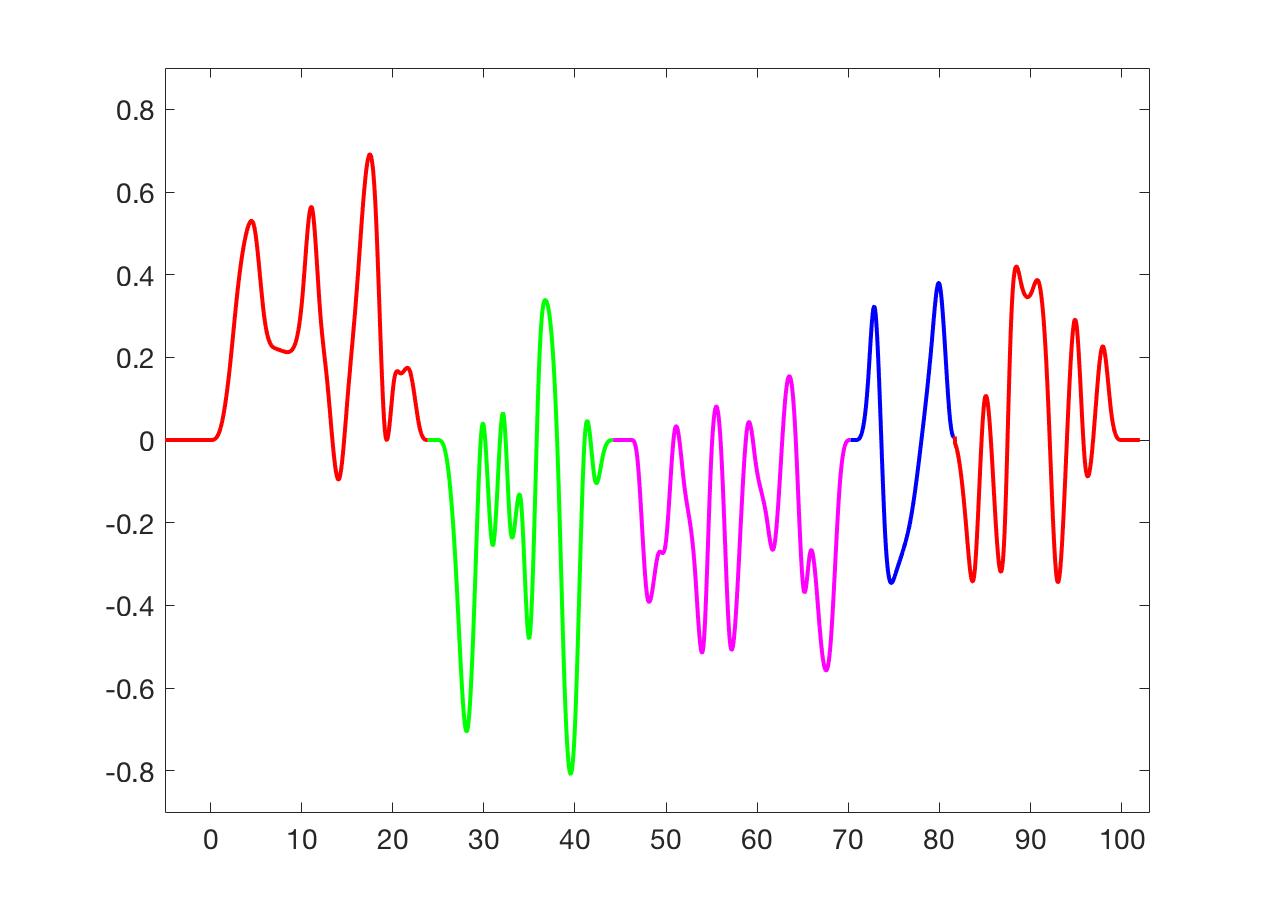}
\includegraphics[width=60mm, height=38mm]{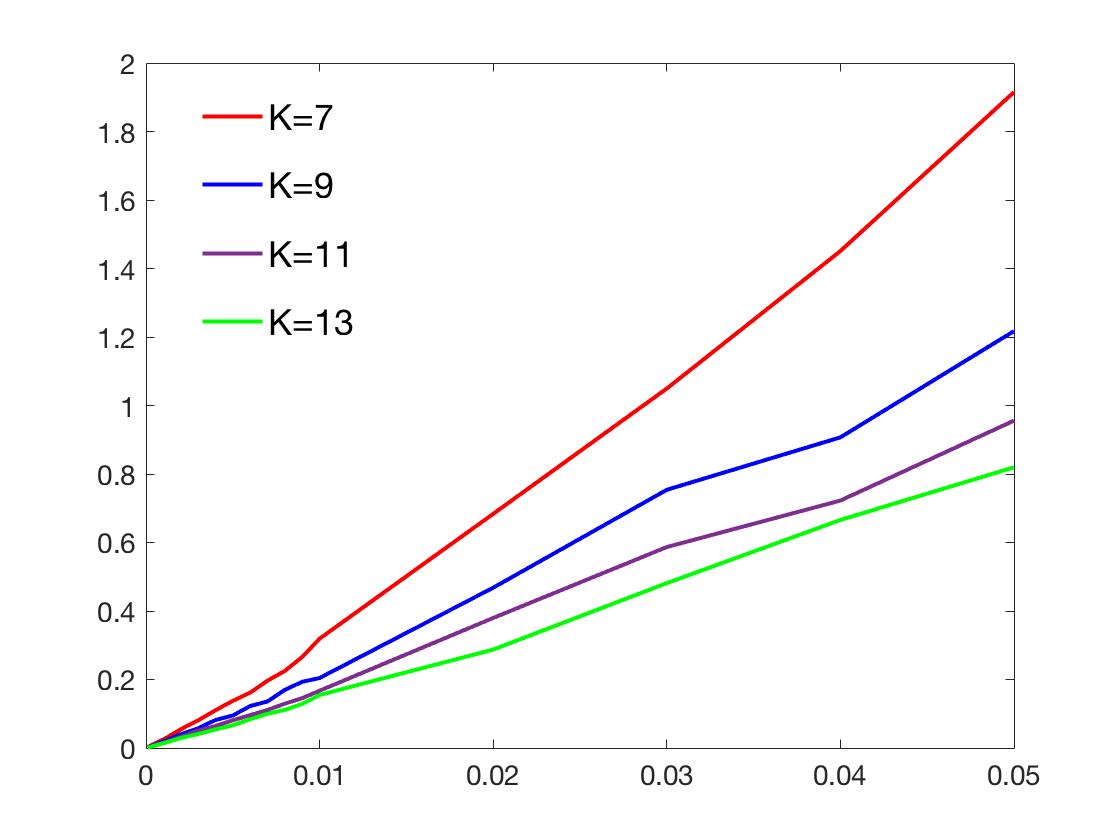}
\end{center}
\begin{center}
\caption{Plotted on the  top left is a  signal $g_\eta$ reconstructed  
 via the proposed MAPS algorithm, while on the top right 
 is the difference $|g_\eta|-|f_o|$
  between magnitude measurements of the reconstructed signal  $g_\eta$ and the original signal  $f_o$ plotted on the left  of Figure
    \ref{spline_separable.fig}. 
 The signal  $h_\eta$ in Theorem \ref{stability.thm} is plotted on the bottom left which has
 five ``islands"
  decomposition
 \eqref{decomposition.def1}, \eqref {decomposition.def2} and \eqref{decomposition.def3}.
 On the bottom right  is the average of maximal reconstruction error  $E_{\eta, K}$ in 200 trials with respect to different noise levels $\eta$ and oversampling rates $K$.
  }\label{spline_quadratic.pic}
\end{center}
\end{figure}

 Performance of the proposed MAPS algorithm depends
 on the noise level $\eta$ and also the oversampling rate $K$, the ratio between
 the density  $K(b-a)/N$ of the  sampling set $\Gamma$ in \eqref{spline.samplingset}  and the  rate  $(b-a)/N$ of innovation of signals in $V(\Phi)$.
 Denote by
 $$E_{\eta, K}:=\||g_\eta|-|f|\|_\infty$$
 the maximal  reconstruction error in   magnitude measurements between the original signal $f$ and the reconstructed signal $g_\eta$ for different noise levels $\eta$ and oversampling rate $K$.
Plotted on the bottom right of  Figure \ref{spline_quadratic.pic} 
are average of the maximal  reconstruction error
 $E_{\eta, K}$
   in 200 trials against the noise level $\eta$ and oversampling rate $K$.  This indicates that the maximal reconstruction error $E_{\eta, K}$
    depends almost linearly on the noise level $\eta$, and  decreases as the oversampling rate $K$ increases,
    cf. \eqref{stability.thm.eq2*} and Theorem \ref{stability.thm}.

\smallskip

Let $D$ be a triangulation composed by  the triangles $T_\theta, \theta\in \Theta$,
and  denote  the set of all inner nodes of the triangulation by $\Lambda$.
In the second simulation, we consider piecewise affine signals
\begin{equation}\label{piecewiseaffine.def}
f(x, y)=\sum_{\lambda\in \Lambda} c_\lambda \phi_\lambda(x, y)\end{equation}
 on   the triangulation $D$,  where the basis signals $\phi_\lambda, \lambda\in \Lambda$ are  piecewise affine  on  triangles  $T_\theta, \theta\in \Theta$ with $\phi_\lambda(\lambda)=1$ and $\phi_\lambda(\lambda')=0$ for all other nodes $\lambda'\ne \lambda$, see  the right image of  Figure \ref{spline_separable.fig}.
From the definition of  basis signals $\phi_\lambda, \lambda\in \Lambda$, a signal  $f$ of the form \eqref{piecewiseaffine.def} has the following interpolation property,
$$f(x, y)=\sum_{\lambda\in \Lambda} f(\lambda) \phi_\lambda(x, y).$$
In the simulation, phaseless samples of a piecewise affine signal  $f$  on a discrete 
 set
$\Gamma=\cup_{\theta\in \Theta} \Gamma_\theta$ are corrupted
by the bounded random noise,
\begin{equation}\label{noisydata2.defcase2}
{ z}_{\eta}({\gamma})= |f({\gamma})|+\eta({ \gamma}), \ {\gamma}\in \Gamma,
\end{equation}
 where $ \eta(\gamma),  \gamma\in \Gamma$,  are randomly selected in the interval $[-\eta, \eta]$ for  some $\eta \ge 0$
 and for every $\theta\in \Theta$, the set $\Gamma_\theta$ contains $7$ points randomly selected inside $T_\theta$.
Shown on the left of Figure
\ref{lambda_set.pic} is a signal
 $g_\eta$  reconstructed from the  noisy phaseless samples  \eqref{noisydata2.defcase2} via the proposed MAPS algorithm, where  $\eta=0.01$, the original piecewise affine signal $f$ is plotted on the right of Figure \ref{spline_separable.fig}, and  the maximal  reconstruction error
 $\||g_\eta|-|f|\|_\infty$  in
    magnitude measurements between the original signal $f$ and the reconstructed signal $g_\eta$  is $0.0360$.

\begin{figure}[H]
\begin{center}
\includegraphics[width=62mm, height=45mm]{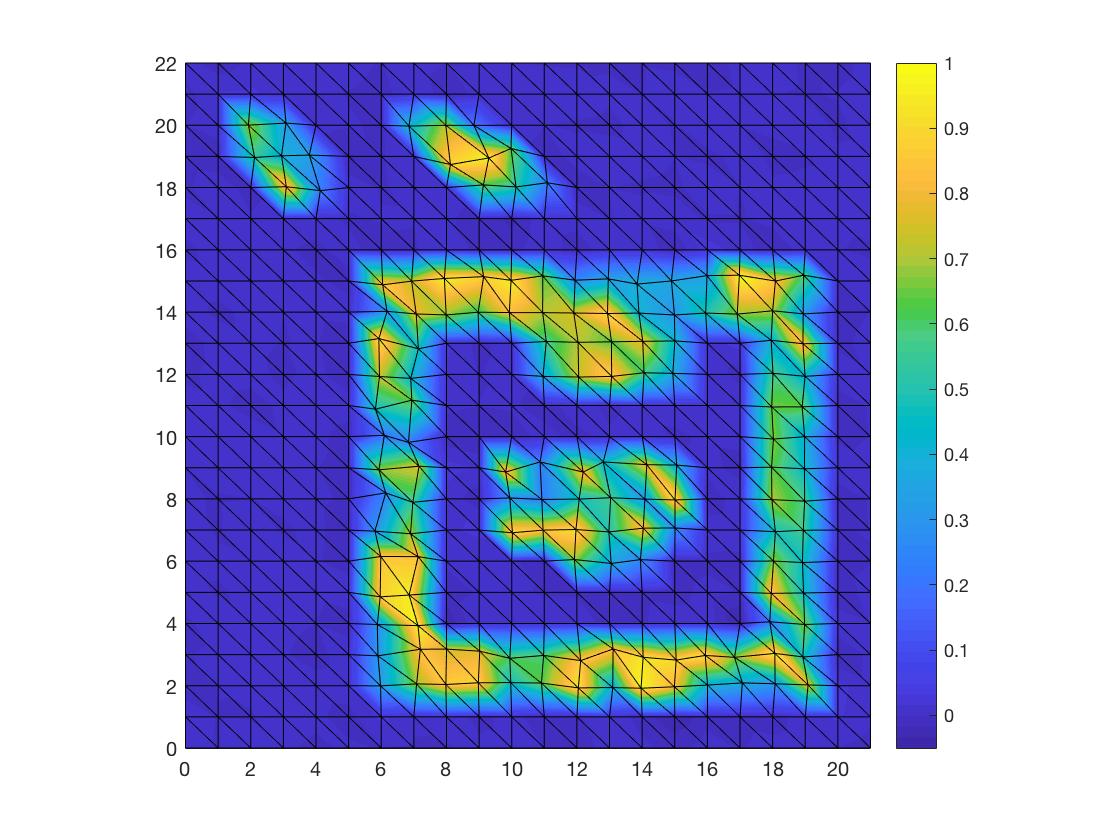}
\includegraphics[width=62mm, height=45mm]{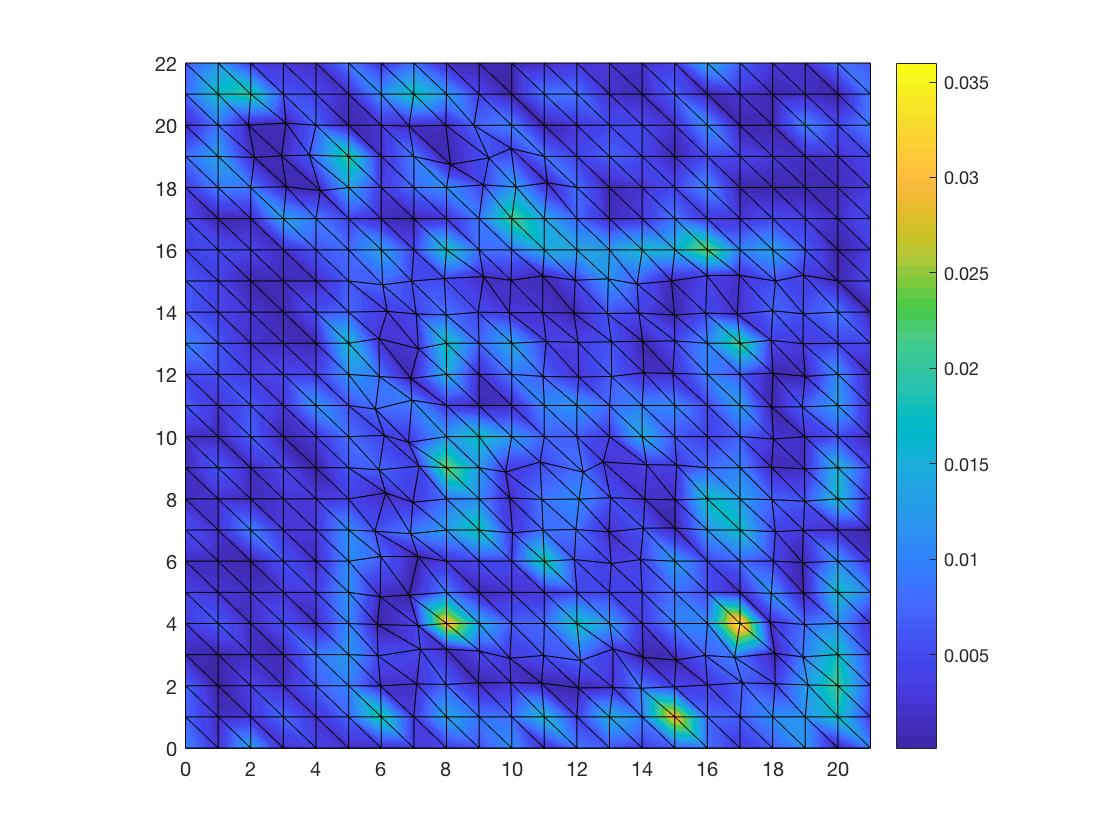}
\caption{Plotted on the left is a reconstructed signal $g_\eta$
via the MAPS algorithm, while on the right is the difference $||g_\eta|-|f||$ between magnitude measurements of the reconstructed signal $g_\eta$ and the original signal $f$ plotted on the right of
 Figure \ref{spline_separable.fig}.  }\label{lambda_set.pic}
\end{center}
\end{figure}

In the simulation, we  consider the performance of the proposed MAPS algorithm to construct piecewise affine  approximation when  the original signal
$f$ of the form \eqref{piecewiseaffine.def}  has evaluations $ f(\lambda), \lambda\in \Lambda$ on their inner nodes being randomly selected
 in $[-1, 1]$. 
Denote by  $g_\eta$ the reconstructed signal from the  noisy phaseless samples  \eqref{noisydata2.defcase2} via the proposed MAPS algorithm
and let  $E_{\eta}:=\||g_\eta|-|f|\|_\infty$  be the maximal reconstruction error  in
 magnitude measurements between the original signal $f$ and the reconstructed signal $g_\eta$ for different noise levels $\eta$.
Shown in Table \ref{separable.tent.table} is the average of maximal reconstruction error $E_{\eta}$ in 200 trials. 
 This confirms  the conclusion in Theorem \ref{stability.thm} that the maximal reconstruction error depends almost linearly on the noise level $\eta\ge 0$.

\begin {table}[H]
\caption {Maximal reconstruction error via the MAPS algorithm}\label{separable.tent.table}
\begin{center}
\begin{tabular}{ |c|c|c|c|c|c|c|c|c|}
\hline
$\eta $
& 0.04 &0.03& 0.02 & 0.01
& 0.008
& 0.004
& 0.002& 0.001 \\
\hline
$E_\eta$ 
& 0.1878 & 0.1366 & 0.0791 &0.0305 
& 0.0226
&	0.0101
  & 0.0050& 0.0025 \\
\hline
\end{tabular}
\end{center}
\end{table}

\begin{appendix}
\section{Density of phaseless sampling sets}
\label{phaseless.appendix}

In the appendix,  we introduce a necessary condition on a discrete set $\Gamma$
such that
${\mathcal M}_{f, \Gamma}={\mathcal M}_f$  for all $f\in V(\Phi)$. We show that that the  density of such a discrete set $\Gamma$
 is no less than the innovation rate of signals  in $V(\Phi)$, see Theorem \ref {density.thm} and Corollary
 \ref{density.cor2}.

\begin{theorem}\label{density.thm}  Let  the domain $D$, the generator
$\Phi: = (\phi_\lambda)_{\lambda\in \Lambda}$,  
the family   ${\mathcal T}=\{T_\theta, \theta\in \Theta\}$ of open sets and the linear space $V(\Phi)$  be as in Theorem \ref{pr.thm}.
If
${\mathcal M}_{f, \Gamma}={\mathcal M}_f$ for all $f\in V(\Phi)$ with ${\mathcal M}_f=\{\pm f\}$, then
\begin{equation}  \label{density.thm.eq2}
D_+(\Gamma)\ge  D_+(\Lambda).
\end{equation}
\end{theorem}

\begin{proof}
Take $x\in D$ and $r\ge r_0$.  By \eqref{density.thm.eq1} and \eqref{generator.assumption.eq1}, 
it suffices to prove that 
\begin{equation}\label{density.thm.pf.eq2}
\#(\Gamma\cap B(x, r))\ge \#  (\Lambda\cap B(x, r-r_0)).
\end{equation}
Assume, on the contrary,  that \eqref{density.thm.pf.eq2} does not hold. Then
there exists a nonzero vector $(d_\lambda)_{\lambda\in \Lambda\cap B(x, r-r_0)}$ such that
\begin{equation} \sum_{\lambda\in B(x, r-r_0)\cap \Lambda} d_\lambda \phi_\lambda(\gamma)=0, \ \gamma\in \Gamma\cap B(x, r).
\end{equation}
Recall that $\phi_\lambda, \lambda\in \Lambda$, are supported in $B(\lambda, r_0)$ by Assumption \ref{generator.assumption}.
Hence
\begin{equation}
\sum_{\lambda\in B(x, r-r_0)\cap \Lambda} d_\lambda \phi_\lambda(\gamma)=0, \ \gamma\in \Gamma\backslash B(x, r).
\end{equation}
Therefore
the set
$$W=\Big\{f:=\sum_{\lambda\in \Lambda\cap B(x, r-r_0)} c_\lambda \phi_\lambda: \   f(\gamma)=0,  \ \gamma\in \Gamma\Big\}\subset V(\Phi)$$
contains nonzero signals.
Take a nonzero signal $f\in W$. By  Theorem \ref{decomposition.thm00}, $f=\sum_{i\in I} f_i$  for some
 nonzero  signals   $f_i\in V(\Phi), i\in I$, such that ${\mathcal M}_{f_i}=\{\pm f_i\}, i\in I$, and
$f_if_i'=0$ for all distinct $i, i'\in I$. This together with $f\in W$ implies that
$f_i(\gamma)=0$ for all $\gamma\in \Gamma$ and  $i\in I$.   Hence $0\in {\mathcal M}_{f_i, \Gamma}, i\in I$,
which contradicts with ${\mathcal M}_{f_i, \Gamma}= {\mathcal M}_{f_i}=\{\pm f_i\}, i\in I$.
\end{proof}

From the above argument, we have the following result without the assumption on the family ${\mathcal T}$ of open sets in Theorem \ref{density.thm}.

\begin{corollary}\label{density.cor2}
Let  the domain $D$ and  the generator
$\Phi: = (\phi_\lambda)_{\lambda\in \Lambda}$ satisfy Assumptions \ref{domain.assumption} and \ref{generator.assumption} respectively,
and 
define the linear space $V(\Phi)$  by \eqref{Vphi.def}.
If $\Gamma$ is a discrete set with ${\mathcal M}_{f, \Gamma}={\mathcal M}_f$ for all $f\in V(\Phi)$, then
 $D_+(\Gamma)\ge  D_+(\Lambda)$.
 \end{corollary}

We finish this appendix with a remark that the low bound in \eqref{density.thm.eq2} can be reached when  the generator $\Phi=(\phi_\lambda)_{\lambda\in \Lambda}$
satisfies that
\begin{equation}
S_\Phi(\lambda, \lambda')=\emptyset \ {\rm for \ all\ distinct} \ \lambda, \lambda'\in \Lambda.
\end{equation}
As in this case, a signal $f\in V(\Phi)$ is nonseparable if and only if $f=c_\lambda \phi_\lambda$ for some $\lambda\in \Lambda$. Thus the set $\Gamma=\{ a(\lambda), \lambda\in \Lambda\}$ is a phaseless sampling set whose upper  density is the same as the rate of innovation, where
$a(\lambda), \lambda\in \Lambda$, are chosen so that $\phi_\lambda(a(\lambda))\ne 0$.
\end{appendix}

\smallskip

{\bf Acknowledgment}. The authors thank Professor Ingrid Daubechies for her valuable comments and suggestions in the
preparation of the paper. 

\end{document}